\documentclass[acmsmall,nonacm]{acmart}

\acmJournal{PACMPL}
\acmVolume{1}
\acmNumber{CONF} 
\acmArticle{1}
\acmYear{2018}
\acmMonth{1}
\acmDOI{} 
\startPage{1}

\setcopyright{none}

\usepackage{booktabs}   
\usepackage{subcaption} 

\usepackage{stel-common}

\newcommand{\by}[1]{\text{/$\mspace{-2mu}$/~#1}}		       

\usepackage{bm} %

\providecommand{\catname}{\mathbf} 
\providecommand{\clsname}{\mathcal}
\providecommand{\oname}[1]{{\operatorname{\mathsf{#1}}}}

\def\defcatname#1{\expandafter\def\csname B#1\endcsname{\catname{#1}}}
\def\defcatnames#1{\ifx#1\defcatnames\else\defcatname#1\expandafter\defcatnames\fi}
\defcatnames ABCDEFGHIJKLMNOPQRSTUVWXYZ\defcatnames

\def\defclsname#1{\expandafter\def\csname C#1\endcsname{\clsname{#1}}}
\def\defclsnames#1{\ifx#1\defclsnames\else\defclsname#1\expandafter\defclsnames\fi}
\defclsnames ABCDEFGHIJKLMNOPQRSTUVWXYZ\defclsnames

\def\defbbname#1{\expandafter\def\csname BB#1\endcsname{{\bm{\mathsf{#1}}}}}
\def\defbbnames#1{\ifx#1\defbbnames\else\defbbname#1\expandafter\defbbnames\fi}
\defbbnames ABCDEFGHIJKLMNOPQRSTUVWXYZ\defbbnames

\def\Set{\catname{Set}}

\providecommand{\argument}{\operatorname{-\!-}}

\DeclareOldFontCommand{\bf}{\normalfont\bfseries}{\mathbf}
\providecommand{\Id}{\operatorname{Id}}

\providecommand{\id}{\mathsf{id}}

\providecommand{\comp}{\mathbin{\circ}}

\providecommand{\xto}[1]{\,\xrightarrow{#1}\,}

\providecommand{\dar}{\kern-1.2pt\operatorname{\downarrow}}	
\providecommand{\uar}{\kern-1.2pt\operatorname{\uparrow}}	

\providecommand{\fst}{\oname{fst}}
\providecommand{\snd}{\oname{snd}}
\providecommand{\pr}{\oname{pr}}

\providecommand{\brks}[1]{\langle #1\rangle}

\providecommand{\inl}{\oname{inl}}
\providecommand{\inr}{\oname{inr}}

\DeclareSymbolFont{Symbols}{OMS}{cmsy}{m}{n}
\DeclareMathSymbol{\iobj}{\mathord}{Symbols}{"3B}
\providecommand{\curry}{\oname{curry}}

\providecommand{\ev}{\oname{ev}}

\usepackage{stmaryrd}

\providecommand{\by}[1]{\text{/\!\!/~#1}}			             %
\providecommand{\pacman}[1]{}					                     %

\newcommand{\undefine}[1]{\let #1\relax}					                       %

\providecommand{\mone}{{\text{\kern.5pt\rmfamily-}\mathsf{\kern-.5pt1}}}

\providecommand{\smin}{\smallsetminus}

\makeatletter
\def\mfix#1{\oname{#1}\@ifnextchar\bgroup\@mfix{}}	       %
\def\@mfix#1{#1\@ifnextchar\bgroup\mfix{}}			           %
\makeatother

\providecommand{\case}[3]{\mfix{case}{\mathbin{}#1}{of}{#2}{\kern-1pt;}{\mathbin{}#3}}

\usepackage{dsfont}

\DeclareMathSymbol{\mathinvertedexclamationmark}{\mathord}{operators}{'074}
\DeclareMathSymbol{\mathexclamationmark}{\mathord}{operators}{'041}
\makeatletter
\newcommand{\raisedmathinvertedexclamationmark}{%
  \mathord{\mathpalette\raised@mathinvertedexclamationmark\relax}%
}
\newcommand{\raised@mathinvertedexclamationmark}[2]{%
  \raisebox{\depth}{$\m@th#1\mathinvertedexclamationmark$}%
}
\makeatother

\newcommand{\HO}{\mathcal{HO}}

\newcommand{\under}[1]{\lvert#1\rvert}
\newcommand{\SKI}{\mathrm{SKI}}

\newcommand{\ap}{{\mathrm{ap}}}
\newcommand{\var}{\mathsf{var}}
\newcommand{\cn}{{\mathrm{cn}}}
\newcommand{\cv}{{\mathrm{cv}}}

\newcommand{\Sigmas}{\Sigma^{\star}}

\newcommand{\ar}{\mathsf{ar}}
\newcommand{\NT}{\mathrm{Nat}}

\newcommand{\epito}{\twoheadrightarrow}

\newcommand{\seq}{\subseteq}
\newcommand{\ol}{\overline}

\newcommand{\beh}{{\mathsf{beh}}}
\providecommand{\C}{}
\providecommand{\D}{}

\renewcommand{\C}{{\mathbb{C}}}
\renewcommand{\D}{{\mathbb{D}}}

\renewcommand{\id}{{\mathsf{id}}}
\renewcommand{\Nat}{\mathds{N}}
\renewcommand{\NT}{\mathsf{Nat}}

\newcommand{\f}{\oname{f}}

\newcommand{\takeout}[1]{\empty}

\newcommand{\ini}{\iota}
\newcommand{\ter}{\tau}

\DeclareMathOperator{\Coalg}{\mathsf{Coalg}}
\DeclareMathOperator{\Alg}{\mathsf{Alg}}

\renewcommand{\rho}{\varrho}

\newcommand{\opp}{\mathsf{op}}

\newcommand{\mypowfin}{\mathscr{P}_\omega}

\newcommand{\pullbackangle}[2][]{\arrow[phantom,to path={
                     -- ($ (\tikztostart)!1cm!#2:([xshift=8cm]\tikztostart) $)
                        node[anchor=west,pos=0.0,rotate=#2,
                        inner xsep = 0]
                        {\begin{tikzpicture}[minimum
                        height=1mm,baseline=0,#1]
    \draw[-] (0,0) -- (.5em,.5em) -- (0,1em);
                        \end{tikzpicture}}}]{}}

\usepackage[capitalize,noabbrev,nameinlink]{cleveref} 

\usepackage{enumitem}
\setlist[enumerate,1]{label=(\arabic*),font=\normalfont,align=left,leftmargin=0pt,labelindent=0pt,listparindent=\parindent,labelwidth=0pt,itemindent=!,topsep=3pt,parsep=0pt,itemsep=3pt,start=1}
\setlist[enumerate,2]{label=(\alph*),font=\normalfont,labelindent=*,leftmargin=*,start=1}
\setlist[itemize]{labelindent=*,leftmargin=*}
\setlist[description]{labelindent=*,leftmargin=*,itemindent=-1 em}

\usepackage{ifdraft}

\renewcommand{\comp}{\cdot}
\renewcommand{\c}{\colon}

\usepackage{seqsplit}
\usepackage{xstring}
\usepackage{xcolor}
\usepackage{bbding}


\usepackage{graphicx,csquotes}
\hypersetup{
  final=true,
  bookmarks=true,
  colorlinks=true,
  hidelinks,
  linkcolor=black,
  citecolor=black,
  filecolor=black,
  urlcolor=black,
  pdftitle={Higher-Order Abstract GSOS},
  pdfauthor={Sergey Goncharov, Stefan Milius, Lutz Schröder, Stelios Tsampas, Henning Urbat},
  pdfkeywords={Abstract GSOS, Categorical Semantics, Higher-order reasoning},
  pdfduplex={DuplexFlipLongEdge},
}

\tikzstyle{shiftarr}=[
        rounded corners,%
        to path={--([#1]\tikztostart.center)
                     -- ([#1]\tikztotarget.center) \tikztonodes
                     -- (\tikztotarget)},
]

\tikzset{
    commutative diagrams/.cd,
    arrow style=tikz,
    diagrams={>=stealth},
    row sep=large,
    column sep = huge
}

\usetikzlibrary{decorations.pathmorphing}
\ifdraft{%
  \makeatletter
  \setcounter{tocdepth}{2}
  \makeatother
}{}

\usepackage[footnote,marginclue,nomargin]{fixme}
\FXRegisterAuthor{hu}{ahu}{HU} 
\FXRegisterAuthor{sm}{asm}{SM} 
\FXRegisterAuthor{st}{ast}{ST} 
\FXRegisterAuthor{ls}{als}{LS} 
\FXRegisterAuthor{sg}{asg}{SG} 

\usepackage{xspace}

\theoremstyle{definition}
\newtheorem{rem}[theorem]{Remark} 

\DeclareMathOperator{\parto}{\rightrightarrows}
\numberwithin{equation}{section}

\newcommand{\xra}[1]{\xrightarrow{~#1~}}
\renewcommand{\xto}{\xra}
\let\xmpsto=\xmapsto
\renewcommand{\xmapsto}[1]{\xmpsto{~#1~}}

\newcommand{\V}{\mathcal{V}}
\newcommand{\Pt}{V}
\newcommand{\iotaq}{\iota_\sim}

\renewcommand{\Nat}{\mathbb{N}}

\newcommand{\fset}{\mathbb{F}}
\newcommand{\vcat}{\set^{\fset}}
\newcommand{\alg}[1]{\mathbf{Alg}(#1)}

\newcommand{\mon}{\bullet}

\newcommand{\gcat}{\mathbb{C}}

\providecommand{\DiNat}{\mathsf{DiNat}}
\newcommand{\mS}{{\mu\Sigma}}
\newcommand{\mSq}{{\mS_{\sim}}}
\newcommand{\BmS}[1]{B(\mS,#1)}
\newcommand{\BmSf}[1]{B(\id,#1)}
\DeclareMathOperator{\iter}{\oname{it}}
\DeclareMathOperator{\coiter}{\oname{coit}}

\newcommand{\app}{\,}

\newcommand{\skitermu}{\Lambda_{u}}
\newcommand{\finalc}{Z}

\theoremstyle{definition}
\newtheorem{notation}[theorem]{Notation}
\newtheorem{assumptions}[theorem]{Assumptions}

\begin{document}

\title[Towards a Higher-Order Mathematical Operational Semantics]{Towards a Higher-Order Mathematical Operational Semantics}         


\author{Sergey Goncharov}
\orcid{nnnn-nnnn-nnnn-nnnn}             
\affiliation{
  \institution{Friedrich-Alexander-Universität Erlangen-Nürnberg}            
  \country{Germany}                    
}
\email{sergey.goncharov@fau.de}          

\author{Stefan Milius}\authornote{Funded by the Deutsche Forschungsgemeinschaft (DFG, German
  Research Foundation) -- project number 470467389}
\orcid{0000-0002-2021-1644}             
\affiliation{
  \institution{Friedrich-Alexander-Universität Erlangen-Nürnberg}            
  \country{Germany}                    
}
\email{stefan.milius@fau.de}          

\author{Lutz Schröder}
\authornote{Funded by the Deutsche Forschungsgemeinschaft (DFG, German
  Research Foundation) -- project number 419850228}  
\orcid{nnnn-nnnn-nnnn-nnnn}             
\affiliation{
  \institution{Friedrich-Alexander-Universität Erlangen-Nürnberg}            
  \country{Germany}                    
}
\email{lutz.schroeder@fau.de}          

\author{Stelios Tsampas}
\authornote{Funded by the Deutsche Forschungsgemeinschaft (DFG, German
  Research Foundation) -- project number 419850228}  
\orcid{nnnn-nnnn-nnnn-nnnn}             
\affiliation{
  \institution{Friedrich-Alexander-Universität Erlangen-Nürnberg}            
  \country{Germany}                    
}
\email{stelios.tsampas@fau.de}          

\author{Henning Urbat}
\authornote{Funded by the Deutsche Forschungsgemeinschaft (DFG, German
  Research Foundation) -- project number 470467389}  
\orcid{nnnn-nnnn-nnnn-nnnn}             
\affiliation{
  \institution{Friedrich-Alexander-Universität Erlangen-Nürnberg}            
  \country{Germany}                    
}
\email{henning.urbat@fau.de}          

\begin{abstract}
Compositionality proofs in higher-order languages are notoriously involved, and general semantic frameworks guaranteeing compositionality are hard to come by. In particular, Turi and Plotkin's bialgebraic abstract GSOS framework, which has been successfully applied to obtain off-the-shelf compositionality results for first-order languages, so far does not apply to higher-order languages. In the present work,  we develop a theory of abstract GSOS specifications for higher-order languages, in effect transferring the core principles of Turi and Plotkin's framework to a higher-order setting. In our theory, the operational semantics of higher-order languages is represented by certain dinatural transformations that we term \emph{pointed higher-order GSOS laws}. We give a general compositionality result that applies to all systems specified in this way and discuss how compositionality of the SKI calculus and the $\lambda$-calculus w.r.t.\ a strong variant of Abramsky's applicative bisimilarity are obtained as instances.
\end{abstract}

\ccsdesc[500]{Theory of computation~Categorical semantics}
\ccsdesc[500]{Theory of computation~Operational semantics}

\keywords{Abstract GSOS, Categorical semantics, Higher-order reasoning}

\maketitle

\section{Introduction}
\label{sec:intro}

Turi and Plotkin's framework  of \emph{mathematical operational
  semantics}~\cite{DBLP:conf/lics/TuriP97} elucidates the operational semantics of programming languages, and guarantees 
compositionality of programming language semantics in all cases that it covers.
In this framework, operational semantics are presented as distributive laws,
varying in complexity, of a monad over a comonad in a suitable category. An
important example is that of \emph{GSOS laws}, i.e. natural transformations of type
\begin{equation}\label{eq:rho}
  \rho_X \colon \Sigma (X \product BX) \to B\Sigmas X, 
\end{equation}
with functors $\Sigma, B \c \gcat \to \gcat$  respectively specifying the
\emph{syntax} and \emph{behaviour} of the system at hand. The idea is that a GSOS law  represents a set of inductive transition rules that specify how programs are run. For example, the choice of $\gcat = \Set$
and $B = (\mypowfin)^{L}$, where $\mypowfin$ is the finite powerset
functor and $L$ a set of transition labels, leads to the
well-known GSOS rule format for specifying labelled transition
systems~\cite{DBLP:journals/jacm/BloomIM95}. In fact, Turi and Plotkin
show that GSOS laws of a (polynomial) endofunctor
$\Sigma \c \Set \to \Set$ over $(\mypowfin)^{L}$ correspond precisely
to GSOS specifications with term constructors given by $\Sigma$ and
terms emitting labels from $L$. For that reason, Turi and Plotkin's framework is
often referred to simply as \emph{abstract GSOS}.

The semantic interpretation of GSOS laws is conveniently presented in a bialgebraic setting (cf.\ \Cref{sec:abstract-gsos}): Every GSOS law $\rho$  \eqref{eq:rho} canonically induces a bialgebra
\[ \Sigma(\mS)\xto{\ini} \mS \xto{\gamma} B(\mS) \]
on the object $\mS$ of programs freely generated by the syntax functor $\Sigma$,
where the algebra structure~$\ini$ inductively constructs programs and the
coalgebra structure $\gamma$ describes their one-step behaviour according to the
given law $\rho$. The above bialgebra is thus the \emph{operational model} of
$\rho$.
Dually, its \emph{denotational model} is a bialgebra
\[ \Sigma(\nu B)\xto{\alpha} \nu B \xto{\tau} B(\nu B),\]
which extends the final coalgebra $\nu B$ of the behaviour functor $B$ (to be
thought of as the domain of abstract program behaviours). Both the operational
and the denotational model can be characterized by universal properties, namely
as the \emph{initial $\rho$-bialgebra} and the \emph{final $\rho$-bialgebra}.
This immediately entails a key feature of abstract GSOS: The semantics is
automatically \emph{compositional}, in that
behavioural equivalence (e.g.\ bisimilarity) of programs is a congruence with respect to the operations of the
language. The bialgebraic framework has been used widely to establish
further correspondences and obtain compositionality
results~\cite{56f40c248cb44359beb3c28c3263838e,
  DBLP:conf/fossacs/KlinS08, DBLP:conf/lics/FioreS06,
  DBLP:journals/tcs/MiculanP16, DBLP:conf/fscd/0001MS0U22}.

As a first step towards extending their framework to languages with
\emph{variable binding}, such as the
$\pi$-calculus~\cite{DBLP:journals/iandc/MilnerPW92a} and the
$\lambda$-calculus, Fiore, Plotkin and
Turi~\shortcite{DBLP:conf/lics/FiorePT99} use the theory of
\emph{presheaves} to establish an abstract categorical foundation of
syntax with variable binding, and develop a theory of capture-avoiding
substitution in this abstract setting. Based on these foundations,
the semantics of \emph{first-order} languages with variable binding,
more precisely that of the $\pi$-calculus and value-passing
CCS~\cite{DBLP:books/daglib/0067019}, is formulated in terms of
GSOS laws on categories of presheaves~\cite{DBLP:conf/lics/FioreT01}.

However, the question of the mathematical operational semantics of the
$\lambda$-calculus, or generally that of higher-order languages, still
remains a well-known issue in the literature (see e.g.~the
introductory paragraph by
\citet{DBLP:journals/lmcs/HirschowitzL22}). Indeed, in order to give
the semantics of a higher-order language in terms of some sort of a
distributive law of a syntax functor over some choice of a behaviour
functor, one needs to overcome a number of fundamental problems. For
instance, for a generic set $X$ of programs, the most obvious set of
``higher-order behaviours over $X$'' would be $X^{X}$, the set of
functions $f \c X \to X$ that expect an input program in~$X$ and
produce a new program. Of course, the assignment $X \mapsto X^{X}$ is
not functorial in $X$ but bifunctorial; more precisely, the bifunctor
$B(X,Y) = Y^{X} \c \Set^{\opp} \product \Set \to \Set$ is of mixed
variance.  Working with mixed-variance bifunctors as a basis for
higher-order behaviour makes the situation substantially more complex
in comparison to Turi and Plotkin's original setting. In particular,
natural transformations alone will no longer suffice as the technical
basis of a framework involving mixed-variance functors, and it is not
a priori clear what the right notion of coalgebra for a mixed-variance
functor should be. In this paper, we address these issues, with a view
to obtaining a general congruence result.


\paragraph{Contributions} We develop a theory of abstract GSOS for
higher-order languages, essentially extending Turi and Plotkin's original
framework. We model higher-order behaviours abstractly in terms of
syntax endofunctors of the form $\Sigma=V+\Sigma'\c \C\to \C$, for a choice of an object $V \in \gcat$ to be thought of as an object of variables, and behaviour
bifunctors $B \c \gcat^{\opp} \product \gcat \to \gcat$. The key concept introduced in our paper is that of a \emph{$\Pt$-pointed higher-order GSOS law}: a family of morphisms
\[
  \rho_{X,Y} \c \Sigma(jX \product B(jX,Y)) \to B(jX,\Sigmas(jX + Y)),
\]
\emph{dinatural} in $X \in \Pt/\gcat$ and \emph{natural} in
$Y \in \gcat$, with $j \c \Pt/\gcat \to \gcat$ denoting
the forgetful functor from the coslice category $\Pt/\gcat$ to $\gcat$.
We show how each $\Pt$-pointed higher-order GSOS law inductively
determines an operational semantics given by a morphism
\[\gamma \c \mS \to B(\mS,\mS)\] where
$\mS$ is initial algebra for the syntax endofunctor $\Sigma \c \gcat \to \gcat$. In fact, much in analogy to the
first-order case, we introduce a notion of \emph{higher-order bialgebra} for
the given law~$\rho$, and show that~$\gamma$ extends to the initial such bialgebra.

From a coalgebraic standpoint, the operational semantics
$\gamma \c \mS \to B(\mS,\mS)$ is a coalgebra for the restricted
\emph{endofunctor} $B(\mS, \argument) \c \gcat \to \gcat$. Our
semantic domain of choice is the final $B(\mS, \argument)$-coalgebra
$(Z,\zeta)$, the object of behaviours determined by the functor
$B(\mS, \argument)$. We obtain a morphism
$\coiter \gamma \c \mS \to Z$ by coinductively extending $\gamma$;
that is, we take the unique coalgebra morphism into the final coalgebra:
\begin{equation*}
  \begin{tikzcd}
    \mS
    \arrow[r, "\gamma"]
    \ar[dashed]{d}[swap]{\coiter \gamma}
    &
    B(\mS, \mS)
    \arrow[d, "B(\mS{,} \coiter \gamma)"]
    \\
    \finalc
    \arrow[r, "\zeta"]
    & B(\mS, \finalc)
  \end{tikzcd}
\end{equation*}
Importantly, and in sharp contrast to first-order abstract GSOS, the final coalgebra
$(Z,\zeta)$ generally does \emph{not} extend to a final higher-order bialgebra; in fact, a final bialgebra usually fails to exist (see \Cref{ex:final-bialgebra}). As a consequence, it turns out that proving compositionality in our higher-order setting is more
challenging, requiring entirely different techniques and additional assumptions on the base category $\C$ and the functors $\Sigma$ and $B$.

Specifically, we investigate higher-order GSOS laws in a \emph{regular} category $\C$. As our main compositionality result, we show that the kernel pair of
$\coiter \gamma \c \mS \to \finalc$ (which under mild
conditions is equivalent to bisimilarity) is a congruence. We
demonstrate the expressiveness of higher-order GSOS laws by modelling
two important examples of higher-order systems. We draw our first
example, the \emph{SKI combinator calculus}~\cite{10.2307/2370619},
from the world of combinatory logic, which we represent using a higher-order GSOS law on the category of sets. For our second
example we move on to a category of presheaves, on which we model the
call-by-name and the call-by-value $\lambda$-calculus.  In all of these
examples, we show how the induced semantics corresponds to
\emph{strong} variants of \emph{applicative
  bisimilarity}~\cite{Abramsky:lazylambda}. We note that in the case of the
call-by-value $\lambda$-calculus, our coalgebraic version of applicative
bisimilarity is nonstandard as it allows application of functions to arbitrary
closed terms, not just values.

While our framework lays the foundations towards a higher-order
mathematical operational semantics in the style of abstract GSOS, let us
mention two of its current limitations and important directions for
future work. First, our compositionality result is about a rather
fine-grained notion of program equivalence, viz.\ coalgebraic
behavioural equivalence, and our ultimate goal is to reason about
\emph{weak} coalgebraic bisimilarity, e.g.\ standard applicative
bisimilarity for $\lambda$-calculi. Second, we presently do not
incorporate effects such as states to our setting. See
\Cref{sec:concl} for more details.

\paragraph{Organization}
In~\Cref{sec:prelim} we provide a brief introduction to the core
categorical concepts that are used throughout this paper. Moving on,
in~\Cref{sec:sets} we discuss examples from combinatory logic and
present a basic rule format for higher-order languages that
illustrates the principles behind our approach.  \Cref{sec:hogsos} is
where we define our notion of pointed higher-order GSOS law and
prove our main compositionality result
(\Cref{th:main}). In~\Cref{sec:lam} we implement the call-by-name and
call-by-value $\lambda$-calculus in our abstract framework.  We
conclude the paper with a summary of the contributions and a brief
discussion on potential avenues for future work in \Cref{sec:concl}.

\paragraph{Related work}

Formal reasoning on higher-order languages is a long-standing research
topic
(e.g.~\cite{3720,Abramsky:lazylambda,DBLP:journals/iandc/AbramskyO93}).
The series of workshops on \emph{Higher Order Operational Techniques
  in Semantics}~\cite{10.5555/309656} played an important role in
establishing the so called \emph{operational methods} for higher-order
reasoning. The two most important such methods are \emph{logical
  relations}~\cite{tait1967intensional, DBLP:journals/iandc/Statman85,
  DBLP:journals/iandc/OHearnR95, DBLP:journals/corr/abs-1103-0510} and
\emph{Howe's method}~\cite{DBLP:conf/lics/Howe89,
  DBLP:journals/iandc/Howe96}, both of which remain in use to date.
Other significant contributions towards reasoning on higher-order
languages were made by
\mbox{\citet{DBLP:journals/iandc/Sangiorgi94, DBLP:journals/iandc/Sangiorgi96}} and
\citet{DBLP:conf/lics/Lassen05}. While GSOS-style
frameworks ensure compositionality for free by mere adherence to given
rule formats, both logical relations and Howe's method instead have
the character of robust but inherently complex methods whose
instantiation requires considerable effort.

Recently, notable progress has been made towards generalizing Howe's
method
\cite{DBLP:journals/lmcs/HirschowitzL22,DBLP:conf/lics/BorthelleHL20},
based on previous work on \emph{familial monads} and operational
semantics~\cite{DBLP:journals/pacmpl/Hirschowitz19}. According to the
authors, their approach departs from Turi and Plotkin's bialgebraic
framework exactly because the bialgebraic framework did not cover
higher-order languages at the time. \citet{DBLP:conf/lics/LagoGL17} give a general account of
congruence proofs, and specifically Howe's method, for applicative
bisimilarity for $\lambda$-calculi with algebraic effects, based on
the theory of relators. \citet{DBLP:journals/entcs/HermidaRR14} present a foundational
account of logical relations as \emph{structure-preserving} relations
in a reflexive graph category.

Rule formats like the GSOS format~\cite{DBLP:journals/jacm/BloomIM95}
have been very useful for guaranteeing congruence of bisimilarity at a
high level of generality. However, rule formats for higher-order
languages have been scarce.  An important example is that of
Bernstein's \emph{promoted tyft/tyxt} rule
format~\cite{DBLP:conf/lics/Bernstein98,DBLP:journals/entcs/MousaviR07},
which has similarities to our presentation of combinatory logic
in~\Cref{sec:sets}, but it is unclear whether or not the format has
any categorical representation.  The rule format of
\citet{DBLP:journals/iandc/Howe96} was presented in the context of
Howe's method. A variant of Howe's format was recently developed by
\citet{DBLP:journals/lmcs/HirschowitzL22}.

\section{Preliminaries}\label{sec:prelim}

\subsection{Category Theory}\label{sec:categories}
We assume familiarity with basic notions from category theory
such as limits and colimits, functors, natural transformations, and monads. For the convenience of the reader, we review some terminology and notation used in the paper.

\paragraph{Products and coproducts} Given
 objects $X_1, X_2$ in a category~$\C$, we write $X_1\times X_2$ for their product, with 
projections $\fst\colon X_1 \times X_2 \to X_1$ and $\snd\colon X_1\times X_2\to X_2$. For a pair of
morphisms $f_i\colon Y \to X_i$, $i = 1,2$, we let
$\brks{f_1,f_2}\colon Y \to X_1 \times X_2$ denote the unique induced
morphism. Dually, we write $X_1+X_2$ for the coproduct, with injections $\inl\c X_1\to X_1+X_2$ and $\inr\c X_2\to X_1+X_2$, and $[g_1,g_2]\c X_1+X_2\to Y$ for the copairing of morphisms $g_i\colon X_i\to Y$, $i=1,2$.

\paragraph{Dinatural transformations}
Given functors $F,G\colon \C^\opp\times \C\to \D$, a \emph{dinatural transformation} from $F$ to $G$ is a family of morphisms $\sigma_X\colon F(X,X)\to G(X,X)$ ($X\in \C$) such that for every morphism $f\colon X\to Y$ of~$\C$, the hexagon below commutes:
\[
\begin{tikzcd}
  & F(X,X) \ar{r}{\sigma_X} & G(X,X) \ar{dr}{G(X,f)} & \\
  F(Y,X) \ar{ur}{F(f,X)} \ar{dr}[swap]{F(X,f)} & & & G(X,Y) \\
  & F(Y,Y) \ar{r}{\sigma_Y} & G(Y,Y) \ar{ur}[swap]{G(f,Y)} & 
\end{tikzcd}
\]

\paragraph{Regular categories}
A category is \emph{regular} if (1)~it has finite limits, (2)~for
every morphism $f\c A\to B$, the kernel pair $p_1,p_2\colon E\to A$ of
$f$ has a coequalizer, and (3)~regular epimorphisms are stable under
pullback.  In a regular category, every morphism $f\colon A\to B$
admits a factorization $A\xto{e} C\xto{m} B$ into a regular
epimorphism $e$ followed by a monomorphism $m$: Take~$e$ to be the
coequalizer of the kernel pair of $f$, and $m$ the unique factorizing
morphism. Indeed the main purpose of regular categories is to provide
a notion of image factorization of morphisms that relates to kernels
of morphisms in a similar way as in set theory. Examples of regular
categories include the category $\Set$ of sets and functions, every
presheaf category $\Set^\C$, and every category of algebras over a
signature (see below). In all these cases, regular epimorphisms and
monomorphisms are precisely the (componentwise) surjective and
injective morphisms, respectively.


\paragraph{Algebras} Given an endofunctor $F$ on a category $\gcat$,
an \emph{$F$-algebra} is a pair $(A,a)$ of an object~$A$
(the \emph{carrier} of the algebra) and a morphism $a\colon FA\to A$
(its \emph{structure}). A \emph{morphism} from
$(A,a)$ to an $F$-algebra $(B,b)$ is a morphism $h\colon A\to B$
of~$\gcat$ such that $h\comp a = b\comp Fh$. Algebras for $F$ and
their morphisms form a category $\Alg(F)$, and an \emph{initial}
$F$-algebra is simply an initial object in that category.  We denote
the initial $F$-algebra by $\mu F$ if it exists, and its structure by
$\ini\colon F(\mu F) \to \mu F$. Moreover, we write
$\iter a\colon (\mu F,\ini) \to (A,a)$ for the unique morphism
from~$\mu F$ to the algebra~$(A,a)$.

More generally, a \emph{free $F$-algebra} on an object $X$ of $\C$ is an
$F$-algebra $(F^{\star}X,\iota_X)$ together with a morphism
$\eta_X\c X\to F^{\star}X$ of~$\C$ such that for every algebra $(A,a)$
and every morphism $h\colon X\to A$ in $\C$, there exists a unique
$F$-algebra morphism $h^\star\colon (F^{\star}X,\iota_X)\to (A,a)$
such that $h=h^\star\comp \eta_X$. As usual, the universal property
determines $F^\star X$ uniquely up to isomorphism. If free algebras
exist on every object, then their formation induces a monad
$F^{\star}\colon \C\to \C$, the \emph{free monad} on $F$
\cite{Barr70}. For every $F$-algebra $(A,a)$, we obtain an
Eilenberg-Moore algebra $\hat a \colon F^{\star} A \to A$ as the free
extension of the identity morphism $\id_A\c A\to A$.


The most familiar example of functor algebras are algebras over a
signature.  An \emph{algebraic signature} consists of a set $\Sigma$
of operation symbols together with a map $\ar\colon \Sigma\to \Nat$
associating to every operation symbol $\f$ its \emph{arity}
$\ar(\f)$. Symbols of arity $0$ are called \emph{constants}. Every
signature~$\Sigma$ induces the polynomial set functor
$\coprod_{\f\in\Sigma} (\argument)^{\ar(\f)}$, which we
denote by the same letter $\Sigma$. An algebra for the functor
$\Sigma$ then is precisely an algebra for the signature $\Sigma$,
i.e.~a set $A$ equipped with an operation $\f^A\colon A^n \to A$ for
every $n$-ary operation symbol $\f\in \Sigma$. Morphisms between
$\Sigma$-algebras are maps respecting the algebraic structure.

An equivalence relation ${\sim}\seq A\times A$ on a $\Sigma$-algebra
$A$ is called a \emph{congruence} if, for every $n$-ary operation
symbol $\f\in \Sigma$ and all elements $a_i,b_i\in A$
($i=1,\ldots,n$), one has that
\[ a_i \sim b_i \quad(i=1,\ldots,n) \qquad\text{implies}\qquad
  \f^A(a_1,\ldots,a_n)\sim \f^A(b_1,\ldots,b_n). \] Equivalently, there
exists a morphism $h\colon A\to B$ to some $\Sigma$-algebra $B$
such that $\sim$ equals the \emph{kernel} of $h$, that is,
\[ a\sim b \qquad\text{iff}\qquad h(a)=h(b). \]

Given a set $X$ of variables, the free algebra $\Sigmas X$ is the
$\Sigma$-algebra of terms generated by $\Sigma$ with variables from
$X$. In
particular, the free algebra on the empty set is the initial algebra
$\mu \Sigma$; it is formed by all \emph{closed terms} of the
signature.
For every $\Sigma$-algebra $(A,a)$, the induced
Eilenberg-Moore algebra $\hat a\colon \Sigmas A \to A$ is given by the map evaluating terms over $A$ in the algebra.

\paragraph{Coalgebras} A \emph{coalgebra} for an
endofunctor $F$ on $\gcat$ is a pair $(C,c)$ of an object $C$ (the
\emph{carrier}) and a morphism $c\colon C\to FC$ (its
\emph{structure}). A \emph{morphism} from an $F$-coalgebra
$(C,c)$ to an $F$-coalgebra $(D,d)$ is a morphism
$h\colon C\to D$ of $\C$ such that $Fh\comp c = d\comp h$.
Coalgebras for $F$ and their morphisms form a category $\Coalg(F)$, and a
\emph{final} $F$-coalgebra is a final object in that category. If
it exists, we denote the final $F$-coalgebra by $\nu F$ and its
structure by $\ter\colon \nu F \to F(\nu F)$, and we write
$\coiter c\colon (C,c)\to (\nu F, \ter)$ for the unique
morphism from the coalgebra $(C,c)$ to $\nu F$.

\subsection{Abstract GSOS}\label{sec:abstract-gsos}
Suppose that $\Sigma, B\c \C\to \C$ are endofunctors on a category with finite products, where $\Sigma$ has a free monad $\Sigmas$. Given a GSOS law $\rho$ as in \eqref{eq:rho}, a \emph{$\rho$-bialgebra} $(X,a,c)$ consists of an object $X\in \C$, a $\Sigma$-algebra $a\c \Sigma X\to X$ and a $B$-coalgebra $c\c X\to BX$ such that the left-hand diagram below commutes. A \emph{morphism} from $(X,a,c)$ to a $\rho$-bialgebra $(X',a',c')$ is a $\C$-morphism $h\c X\to X'$ that is both a $\Sigma$-algebra morphism and a $B$-coalgebra morphism, i.e.\ the right-hand diagram commutes:
\[
\begin{tikzcd}
\Sigma X \ar{r}{a} \ar{d}[swap]{\Sigma\langle \id,\,c\rangle} & X \ar{r}{c} & BX \\
\Sigma(X\times BX) \ar{rr}{\rho_X} & & B\Sigmas X \ar{u}[swap]{B\hat a}  
\end{tikzcd}
\qquad\qquad
\begin{tikzcd}
\Sigma X \ar{r}{a} \ar{d}[swap]{\Sigma h} & X \ar{d}{h} \ar{r}{c} & BX \ar{d}{Bh} \\
\Sigma X' \ar{r}{a'} & X' \ar{r}{c'} & BX'  
\end{tikzcd}
\]
The universal property of the initial algebra $(\mS,\iota)$ entails that there
exists a unique $B$-coalgebra structure $\gamma\c \mS\to B(\mS)$ such that
$(\mS,\iota,\gamma)$ is a $\rho$-bialgebra. This is the initial
$\rho$-bialgebra, i.e.\ the initial object in the category of bialgebras and
their morphisms. Dually, if $B$ has a final coalgebra $(\nu B,\tau)$, it
uniquely extends to a $\rho$-bialgebra $(\nu B,\alpha,\tau)$, and this is the
final $\rho$-bialgebra. Thus, both by initiality and finality, we obtain a unique bialgebra morphism
\[ \beh_\rho\c (\mS,\iota,\gamma)\to (\nu B,\alpha,\tau); \]
 we think of $\beh_\rho$ as the map assigning to each program in $\mS$ its abstract behaviour. 
 For $\C=\Set$ and $\Sigma$ a polynomial functor, the fact that $\beh_\rho$ is a
 $\Sigma$-algebra morphism immediately implies that behavioural equivalence,
 namely the relation $\equiv$ on $\mS$ given by $p\equiv q$ iff
 $\beh_\rho(p)=\beh_\rho(q)$, is a congruence. In other words, the semantics
 induced by $\rho$ is \emph{compositional}.


\section{Combinatory logic}
\label{sec:sets}
We ease the reader into our theory of higher-order languages by considering the SKI
calculus \cite{10.2307/2370619}, a combinatory logic that is as expressive as
the untyped $\lambda$-calculus but does not feature variables, thus avoiding
complications arising from binding and substitution. Specifically, we
investigate a variant of the SKI calculus, which we call the $\SKI_{u}$ calculus, featuring
auxiliary operators. It is as expressive as the standard presentation but
semantically simpler. 

\subsection{The $\mathbf{SKI}_{u}$ Calculus}
\label{sec:unary-ski}\label{sec:unary}
 The set $\skitermu$ of $\SKI_u$ terms
is generated by the grammar
%
\[
  \skitermu ::= S \mid K \mid I \mid \skitermu \circ \skitermu \mid
  S'(\skitermu) \mid K'(\skitermu) \mid S''(\skitermu,\skitermu).
\]
Intuitively, the binary operation $\argument \circ \argument$ corresponds to function application; we usually write $s\,t$ for $s\circ t$. The standard combinators (constants) $S$, $K$ and $I$ are meant to represent the ternary function $(s,t,u)\mapsto (s\app u)\app (t\app u)$, the first projection $(s,t)\mapsto s$, and the identity map $s\mapsto s$, respectively. The unary operators
$S'$ and $K'$ capture application of~$S$ and~$K$, respectively, to one
 argument: $S'(t)$ behaves like $S\app t$, and $K'(t)$
behaves like $K\app t$.
Finally, the binary operator $S''$ is meant to
capture application of $S$ to two arguments: $S''(s,t)$ behaves like
$(S\app s)\app t$. In this way, the behaviour of each combinator can
be described in terms of \emph{unary} higher-order functions; for example, the
behaviour of $S$ is that of a function taking a term $t$ to $S'(t)$.

The small-step operational semantics of $\SKI_u$ is given by the rules displayed in 
\Cref{fig:skirules}, where $t,p,p',q$ range over terms in $\skitermu$. The operational semantics determines a labelled transition system $\to~
\subseteq \skitermu \product (\skitermu + \{\_\}) \product \skitermu$ by
induction on the structure of terms in $\skitermu$, with $\{\_\}$ denoting the lack
of a transition label. In this instance the set $\skitermu$ of labels coincides
with the state space of the transition system. Note that every $t\in\skitermu$ either admits a single
unlabelled transition $t\to t'$ or a family of labelled transitions
${(t\xto{s} t_s)_{s\in\skitermu}}$; thus, the transition system is deterministic. The intention is that unlabelled transitions
correspond to \emph{reductions} and that labelled transitions represent
\emph{higher-order behaviour}: supply a term $s$ to produce a new term $t_{s}$.
An important fact towards specifying an abstract format in
\Cref{sec:ho} is that labelled transitions are uniformly defined for every
input $s \in \skitermu$, in that operators make no assumptions on the structure
of $s$.

\begin{figure*}[t]
\begin{gather*}
\frac{}{S\xto{t}S'(t)}
\qquad
\frac{}{S'(p)\xto{t}S''(p,t)}
\qquad
\frac{}{S''(p,q)\xto{t}(p\app t)\app (q\app t)}
\\[1ex]
\frac{}{K\xto{t}K'(t)}
\qquad
\frac{}{K'(p)\xto{t}p}
\qquad
\frac{}{I\xto{t}t}
\qquad
\frac{p\to p'}{p \app q\to p' \app q}
\qquad
\inference[\texttt{app1}]{p\xto{q} p'}{p \app q\to p'}
\end{gather*}
\caption{Operational semantics of the $\SKI_u$ calculus.}
\label{fig:skirules}
\end{figure*}

%
%
%

Like every labelled transition system, $\skitermu$ comes with a notion of \emph{(strong) bisimilarity}. Recall that a relation
$R \subseteq \skitermu \product \skitermu$ is a \emph{bisimulation} for $\to$ if and
only if whenever $p \mathbin{R} q$, then
\begin{enumerate}
\item $p\to p' \implies \exists q'.\, q\to q' \wedge p' \mathbin{R} q'$,
\item $q\to q' \implies \exists p'.\, p\to p' \wedge p' \mathbin{R} q'$,
\item $\forall t\in\skitermu,~p\xto{t} p' \implies \exists q'.\, q\xto{t} q' \wedge p'
  \mathbin{R} q'$,
\item $\forall t\in\skitermu,~q\xto{t} q' \implies \exists p'.\, p\xto{t} p' \wedge p'
  \mathbin{R} q'$.
\end{enumerate}
We write $\sim$ for the greatest $\to$-bisimulation, and call two states $p$ and
$q$ \emph{bisimilar} if $p\sim q$.

The relation $\sim$ identifies combinators that `do the same thing', in that they
produce bisimilar terms given the same input.

\begin{example}
  \label{ex:ski}
The terms $(S\app K) \app I$ and $(S \app K) \app K$ transition as follows:
\[
\arraycolsep=1.4pt
 \begin{array}{ccccccccccc}
    (S \app K) \app I & \to & S'(K)\app I&\to & S''(K,I)&\xto{t}& (K \app t) \app (I \app
    t)&\to& K'(t) \app (I \app t)&\to& t,\\
    (S \app K) \app K&\to& S'(K)\app K&\to& S''(K,K)&\xto{t}& (K \app t) \app (K \app
    t)&\to& K'(t) \app (K \app t)&\to& t.
  \end{array}
\]
  It follows that $(S\app K)\app I \sim (S\app K) \app K$.
\end{example}



The set $\Lambda_u$ of $\SKI_u$-terms forms the initial algebra for
the signature $\Sigma=\{S,K,I,S',K',S'',\circ\}$, and the algebraic
structure respects bisimilarity:
\begin{proposition}[Compositionality of $\SKI_u$]
  \label{prop:skicong1}
  The bisimilarity relation $\sim$ is a congruence.
\end{proposition}
\begin{proof}
  In the following, by a \emph{context} we mean a term
  $C\in \Sigmas\{x\}$ in which the variable $x$ appears exactly
  once. We write $C[p]=C[p/x]\in \skitermu$ for the closed term
  obtained by substituting $p\in \skitermu$ for $x$ in $C$. An
  equivalence relation $R\subseteq\skitermu\times\skitermu$ is a
  congruence if and only if the following relation is contained in $R$: 
  \begin{align*}
    \hat R = \{(C[p],C[q])\in\skitermu\times\skitermu\mid \text{$C$ is
      a context and $p \mathbin{R} q$}\}.
  \end{align*}
  Thus, our task is to prove ${\hat\sim}\subseteq{\sim}$. To this end,
  it suffices to prove that $\hat\sim$ is a bisimulation up to
  transitive closure. This means that for every context~$C$ and
  $p,q\in\skitermu$ such that $p \sim q$,
\begin{itemize}
\item either there exist $p',q'\in \Lambda_u$ such that $C[p]\to p'$, $C[q]\to q'$ and $p' \hat\sim^\star q'$,
\item or for every $t\in\skitermu$, there exist $p',q'\in \Lambda_u$ such that $C[p]\xto{t} p'$, $C[q]\xto{t} q'$
    and~$p' \hat\sim^\star q'$,
\end{itemize}
where $\hat\sim^\star$ denotes the transitive closure of
$\hat\sim$. This then implies that $\hat\sim^\star$ is a bisimulation,
and consequently we obtain ${\hat\sim}\seq {\hat\sim^\star}\seq {\sim}$ because $\sim$ is
the greatest bisimulation.

We proceed by structural induction on $C$. The cases where~$C$ is a
combinator term are straightforward. For instance, for
$C = S''(r,C')$, the property in question can be read from the
diagram
\begin{equation*}
\begin{tikzcd}[column sep=normal, row sep=normal]
S''(r,C'[p])
  \rar[phantom,description,"\hat\sim"]
  \dar["t"'] &
S''(r,C'[q])
  \dar["t"]\\
(r \app t)\app (C'[p] \app t)
  \rar[phantom,description,"\hat\sim^\star"] &
(r \app t)\app (C'[q] \app t)
\end{tikzcd}
\end{equation*}
The other combinators are handled analogously. For transitions of
application terms, we distinguish cases as follows.
\begin{itemize}
\item If $C=C' \app r$ and the transition of $C[p]$ comes from an
  unlabelled transition ${C'[p]\to p'}$, then the transition of $C[p]$
  is $C'[p]\app r\to p' \app r$, and by induction, we have~$q'$ such that
  $C'[q]\to q'$ and $p' \mathbin{\hat\sim^\star} q'$; then $C[q]=C'[q]\app r\to
  q' \app r$ and $p'\app r \mathbin{\hat\sim^\star} q' \app r$.
\item If $C=C'\app r$ and the transition of $C[p]$ comes from a labelled
  transition $C'[p]\xto{r} p'$, then the transition of $C[p]$ is
  $C'[p]\app r\to p'$. By induction, we have~$q'$ such that
  $C'[q]\xto{r} q'$ and $p' \mathbin{\hat\sim^\star} q'$, and then $C'[q]\app r\to
  q'$.
\item If $C=r\app C'$ and the transition of $C[p]$ comes from an unlabelled
  transition $r\to r'$, then we have
  $r \app C'[p]\to r'\app C'[p]$, which completes the case since $r'\app
  C'[p] \mathbin{\hat\sim^\star} r'\app C'[q]$.
\item Finally, suppose that $C=r\app C'$ and the transition of $C[p]$
  comes from a labelled transition of~$r$. According to the rules in
  \Cref{fig:skirules}, for every $t\in\skitermu$ we have
  $r\xto{t} r'[t/x]$ for some term $r'$ with one free variable $x$
  such that $r'$ depends only on $r$ but not on $t$. Hence,
  $r\xto{C'[p]} r' [C'[p]/x]$, $r\app C'[p]\to r' [C'[p]/x]$ and
  similarly $r\app C'[q]\to r' [C'[q]/x]$. Since
  $C'[p]\mathbin{\hat\sim} C'[q]$, we conclude that
  $r'[C'[p]/x] \mathbin{\hat\sim^\star} r' [C'[q]/x]$. (Note that the
  variable $x$ can appear multiple times in $r'$, so we generally do
  not have $r'[C'[p]/x] \mathbin{\hat\sim} r' [C'[q]/x]$. This is the
  reason why we need to work with the transitive closure
  ${\hat\sim^\star}$.) \qedhere
\end{itemize}
\end{proof}

The proof of \Cref{prop:skicong1} is laborious, as it requires tedious
case distinctions and a carefully chosen up-to technique, although the latter
could have been avoided by working with multi-hole contexts. It also
only applies to a specific language, one among many systems exhibiting
higher-order behaviour. In the sequel, we describe an abstract,
categorical representation of such higher-order systems that
guarantees the compositionality of the semantics. In particular, we
shall see that \Cref{prop:skicong1} emerges as an instance of a
general compositionality result (\Cref{th:main}).

\subsection{A Simple Higher-order Abstract Format}
\label{sec:ho}
From a coalgebraic perspective~\cite{DBLP:journals/tcs/Rutten00},
finitely branching labelled transition systems over a set $L$ of labels and
state space $X$ correspond to functions of the form
\[
  h \c X \to (\mypowfin X)^{L},
\]
where $\mypowfin \c \Set \to \Set$ is the finite powerset functor. In
other words, they are $(\mypowfin)^{L}$-coalgebras. The
transition system
$\to~ \subseteq \skitermu \product (\skitermu + \{\_\}) \product
\skitermu$ for the $\SKI_u$-calculus is deterministic, and terms in
$\skitermu$ always have either $t$-labelled transitions for all
$t\in\skitermu$, or an unlabelled transition (but not both). Thus,
the LTS $\to$ can be reformulated as a map of the form
\begin{equation}
  \label{eq:skiutr}
  \gamma_{u} \c \skitermu \to \skitermu +
  \skitermu^{\skitermu},
\end{equation}
where the two summands of the codomain represent unlabelled and labelled transitions, respectively. In other words, $\skitermu$ carries the structure of a coalgebra for the set functor $Y\mapsto Y+Y^{\skitermu}$.
We can abstract away from the set~$\skitermu$ of labels  and
consider $\gamma_{u}$ as a system of the form
\begin{equation*}
  Y \to Y + Y^{X}.
\end{equation*}
For higher-order systems such as $\gamma_{u} \c \skitermu \to \skitermu +
\skitermu^{\skitermu}$, we expect $X = Y$, underlining the fact that inputs
come from the state space of the system. We note that the assignment
\begin{equation}
  \label{eq:simplehigher}
  B_{u}(X,Y) = Y + Y^{X} \c \Set^{\opp} \product \Set \to \Set
\end{equation}
gives rise to a \emph{bifunctor} that is contravariant in $X$ and
covariant in $Y$. On the side of syntax, the signature of
$\SKI_u$ yields an endofunctor $\Sigma \c \Set \to \Set$ given by
\begin{equation*}
  \Sigma X = \coprod\nolimits_{\f\in\{S,K,I,S',K',S'',\circ\}} X^{\ar(\f)},
\end{equation*}
where $\ar(\f)$ is the arity of the operator $\f$ in $\SKI_u$.

Recall that GSOS specifications are in bijective correspondence with natural
transformations
\[
  \Sigma(X \product (\mypowfin X)^{L}) \to (\mypowfin(\Sigmas X))^{L},
\]
i.e. \emph{GSOS laws} of the endofunctor $\Sigma \c \Set \to \Set$ over
$(\mypowfin)^{L}$~\cite{DBLP:conf/lics/TuriP97}. As we are dealing with a behaviour bifunctor $B_u(X,Y)=Y +
Y^{X}$ in lieu of an endofunctor, our first abstract argument is that
specifications of simple combinatory calculi such as the $\SKI_u$ calculus
correspond to families of functions
\begin{equation*}
  \rho_{X,Y} \c \Sigma(X \times B_{u}(X,Y))\to B_{u}(X,\Sigmas(X+Y)),
\end{equation*}
that are \emph{dinatural} in $X$ and \emph{natural} in $Y$. To make things more
instructive and clear, we drive this argument by introducing a simple concrete
rule format for higher-order combinatory calculi with unary operators, which we call the
\emph{$\HO$ rule format}. Let us fix an algebraic signature $\Sigma$
or, equivalently, a polynomial set functor $\Sigma$,
and a countably infinite set of metavariables \[\V = \{x\} + \{x_{1},x_2,\dots\}+ \{y_{1},y_2,\dots\} + \{ y_i^z : i\in \{1,2,3,\ldots\} \text{ and } z\in \{x, x_1,x_2,x_3,\ldots\}\}.\]

\begin{definition}
  \label{def:hoformat}
\begin{enumerate}
\item An \emph{$\HO$ rule} for an $n$-ary operation symbol
$\f\in \Sigma$ is an expression of the form
\begin{equation}\label{eq:rule-red}
  \inference{(x_j\to y_j)_{j\in
      W}\qquad(\goesv{x_{i}}{y^{z}_{i}}{z})_{i\in\{1,\ldots,n\}\smin
      W,\,z \in \{x_1,\ldots,x_n\}}}{\f(x_1,\ldots,x_{n})\to
    t}
\end{equation}
or
\begin{equation}\label{eq:rule-nonred}
  \inference{(x_j\to y_j)_{j\in W}\qquad(\goesv{x_{i}}{y^{z}_{i}}{z})_{i\in\{1,\ldots,n\}\smin
      W,\,z \in \{x,x_1,\ldots,x_n\}}}{\goesv{\f(x_1,\ldots,x_{n})}{t}{x}}
\end{equation}
where $W\seq \{1,\ldots, n\}$, and $t\in \Sigmas \V$ is a term depending only on the variables occurring in the premise; that is, in \eqref{eq:rule-red} the term $t$ can contain the variables $x_i$ ($i=1,\ldots,n$), $y_j$ ($j\in W$), and $y_i^{x_j}$ ($i\in \{1,\ldots,n\}\smin W$, $j=1,\ldots,n$), and in \eqref{eq:rule-nonred} it can additionally contain $x$ and $y_i^x$ ($i\in \{1,\ldots,n\}\smin W$).
\item  An \emph{$\HO$ specification} for $\Sigma$ is a set of $\HO$
  rules such that for
  each $n$-ary operation symbol $\f\in \Sigma$ and each subset $W \seq
  \{1,\ldots,n\}$ there is exactly one rule of the form \eqref{eq:rule-red} or \eqref{eq:rule-nonred} in the set.
\end{enumerate}
\end{definition}

Intuitively, for every given rule the subset $W$ determines which of
the subterms of $\f$ perform a reduction and which exhibit
higher-order behaviour, i.e.~behave like functions. For
$i\in \{1,\dots,n\}\setminus W$, the format dictates that said
functions can be applied to a left-side variable $x_{j}$ or the input
label~$x$, and then the output $x_{i}(x_{j}) = y_{i}^{x_j}$ or
$x_i(x)=y_i^x$ can be used in the conclusion term $t$. The uniformity is apparent:
rules cannot make any assumptions on the input label $x$ or on other
left-side variables that are used as arguments on the premises.

\begin{example}\label{ex:ski-to-ho}
  The rules in \Cref{fig:skirules} form an $\HO$ specification up to renaming variables and adding useless premises. For illustration, let us consider the rule \texttt{app1}. First, using the variables $x_1, x_2, x_1^{x_2}$ instead of $p,q,p'$, the rule can be rewritten as
\[ \inference[\texttt{app1}]{x_1\xto{x_2} x_1^{x_2}}{x_1 \app x_2\to x_1^{x_2}}. \]
This is not yet an $\HO$ rule, since the latter require a complete list of premises. However, by filling in the missing premises for $x_2$ in every possible way we can turn \texttt{app1} into the following two $\HO$ rules, corresponding to the cases $W=\emptyset$ and $W=\{2\}$ in \eqref{eq:rule-red}:
  \begin{gather*}
    \inference[\texttt{app1-a}]{\goesv{x_1}{x_1^{x_1}}{x_1} & \goesv{x_1}{x_1^{x_2}}{x_2} &
      \goesv{x_2}{x_2^{x_1}}{x_1} & \goesv{x_2}{x_2^{x_2}}{x_2}} {\goes{x_1 \app x_2}{x_1^{x_2}}} \\
  \inference[\texttt{app1-b}]{\goesv{x_1}{x_1^{x_1}}{x_1} & \goesv{x_1}{x_1^{x_2}}{x_2} &
     \goes{x_2}{y_2}}{\goes{x_1 \app x_2}{x_1^{x_2}}}
  \end{gather*}
\end{example} 
We are now in a position to show that $\HO$
specifications are in bijection  with (di)natural transformations of a certain shape.

\begin{proposition}
  \label{prop:yon1}
  For every algebraic signature $\Sigma$ and $B_{u}(X,Y)=Y+Y^X$ as defined in
  \eqref{eq:simplehigher}, $\HO$ specifications for $\Sigma$ are in 
  bijective correspondence with families of maps
\begin{equation}
  \label{eq:simpleho}
  \rho_{X,Y} \c \Sigma(X \times B_{u}(X,Y))\to
  B_{u}(X,\Sigmas(X+Y))\qquad (X,Y\in \Set)
\end{equation}
dinatural in $X$ and natural in $Y$.
\end{proposition}

\begin{rem}
  The need for dinaturality  comes from
  the mixed variance of the functor $B_u$, which in turn is caused
  by the fact that variables are used both as states (covariantly) and
  as labels (contravariantly). The role of dinaturality is then the
  same as otherwise played by naturality: It states on an abstract
  level that the rules are parametrically polymorphic, that is, they do not
  inspect the structure of their arguments.

In more technical terms, (di)naturality enables the use of the Yoneda lemma to establish the bijective correspondence of \Cref{prop:yon1}. Explicitly, the bijection maps an $\HO$ specification $\mathcal{R}$ to the family \eqref{eq:simpleho} defined as follows. Given $X,Y\in \Set$ and \[w=\f((u_1,v_1),\ldots, (u_n,v_n))\in \Sigma(X\times B_u(X,Y)),\]
consider the rule in $\mathcal{R}$ matching $\f$ and $W=\{ j\in \{1,\ldots,n\} : v_j\in Y  \}$. If it is of the form \eqref{eq:rule-red}, then 
\[\rho_{X,Y}(w)\in \Sigmas(X+Y)\seq B_u(X,\Sigmas(X+Y)) \]
is the term obtained by taking the term $t$ in the conclusion of \eqref{eq:rule-red} and applying the substitutions
\[ x_i\mapsto u_i~(i\in \{1,\ldots,n\}),\quad y_j\mapsto v_j~(j\in W), \quad y_i^{x_j}\mapsto v_i(u_j)~(i\in \{1,\ldots,n\}\smin W, j\in \{1,\ldots,n\} ).  \]
If the rule is of the form \eqref{eq:rule-nonred}, then 
\[\rho_{X,Y}(w)\in \Sigmas(X+Y)^X\seq B_u(X,\Sigmas(X+Y)) \]
is the map $u\mapsto t_u$, where the term $t_u$ is obtained by taking the term $t$ in the conclusion of \eqref{eq:rule-nonred} and applying the above substitutions along with
\[ x\mapsto u\qquad\text{and}\qquad y_i^{x}\mapsto v_i(u)~(i\in \{1,\ldots,n\}\smin W). \]
\end{rem}

\begin{example}\label{ex:ski-to-rho}
Let $\rho$ be the higher-order GSOS law corresponding to the $\HO$ specification of $\SKI_u$, see \Cref{ex:ski-to-ho}. Given $w = (u_1,v_1)\app (u_2,v_2)\in \Sigma(X\times B_u(X,Y))$ where $v_1\in Y^X$,  one has $\rho_{X,Y}(w)=v_1(u_2)$, according to the rule $\texttt{app1}$.
\end{example}

\subsection{Nondeterministic $\mathbf{SKI_u}$}\label{sec:nd-ski}
Just as the $\lambda$-calculus, combinatory logic can be enriched with other features,
such as nondeterminism, and the theory of applicative bisimulations can be readily
developed for such extensions. For the $\lambda$-calculus this has been pioneered by 
\citet{DBLP:journals/iandc/Sangiorgi94}. For example, consider an extension 
of $\SKI_u$ with a binary operator $\oplus$ representing nondeterministic choice. Its grammar is given by
%
%
\[
  \skitermu^\oplus
  ::=
  S \mid K \mid I \mid
  \skitermu^\oplus \circ \skitermu^\oplus
  \mid
  S'(\skitermu^\oplus)
  \mid
  K'(\skitermu^\oplus)
  \mid
  S''(\skitermu^\oplus,\skitermu^\oplus)
  \mid
  \skitermu^\oplus \oplus \skitermu^\oplus.
\]
On the side of the operational semantics, $\SKI_u^{\oplus}$ has the 
same rules as $\SKI_u$ (see \Cref{fig:skirules}), plus the following ones for resolving nondeterminism:
\begin{align*}
\inference{}{p\oplus q\to p}
&&
\inference{}{p\oplus q\to q}
\end{align*}
This semantics calls for the modification of the behaviour bifunctor $B_u(X,Y)=Y+Y^X$ to
\begin{equation}
  \label{eq:higher-nd}
  B_{u}^{\oplus}(X,Y) = \mypowfin(Y+Y^{X}) \c \Set^{\opp} \product \Set \to \Set,
\end{equation}
where $\mypowfin$ is the finite powerset functor. Sets of nondeterministic transitions rules such as for $\SKI_u^{\oplus}$ then correspond to families of functions
\begin{equation*}
  \rho_{X,Y} \c \Sigma(X \times B_{u}^\oplus(X,Y))\to B_{u}^\oplus(X,\Sigmas(X+Y))
\end{equation*}
that are dinatural in $X$ and natural in $Y$.
In analogy to \Cref{prop:skicong1} we get the following compositionality result:
\begin{proposition}\label{prop:skicong2}
  Bisimilarity is a congruence for the nondeterministic $\SKI_u$ calculus.
\end{proposition}
Rather than giving yet another proof by induction on the syntax, we will derive
the above proposition from our abstract congruence result in \Cref{th:main}. This
highlights the advantage of the genericity achieved by working in a
category-theoretic framework.

\section{Higher-order abstract GSOS}
\label{sec:hogsos}
We are now ready to present our main contribution, a theory of abstract GSOS for
higher-order systems. \Cref{prop:yon1} suggests that we can
abstract away from specific behaviour bifunctors and consider families
of morphisms 
\begin{align}\label{eq:law-specialcase}
\rho_{X,Y}\c\Sigma(X\times B(X,Y))\to B(X,\Sigma^\star (X+Y)),
\end{align}
dinatural in $X$ and natural in $Y$, as an abstract format that is parametric
in the base category $\gcat$ and two functors $\Sigma \c \gcat
\to \gcat$ and $B\colon \C^\opp\times \C\to \C$ representing syntax and behaviour, respectively.

With the developments in \Cref{sec:lam} in mind, we will
actually work with a slightly more general format where $X$ is
required to be a \emph{pointed object}. More precisely, given a fixed
object $\Pt$ of a category~$\gcat$, let~$\Pt/\gcat$ denote the induced coslice
category. Its objects are the \emph{$\Pt$-pointed objects} of~$\gcat$,
that is, pairs $(X,p_X)$ of an object $X\in\gcat$ and a morphism
$p_X\colon\Pt\to X$ of $\gcat$. A morphism from~$(X,p_X)$ to $(Y,p_Y)$
is a morphism $h\colon X\to Y$ of $\C$ such that $h\comp p_X =
p_Y$. The idea is that the object $\Pt$ represents variables,~$X$~is a
set of program terms in free variables from $\Pt$, and the map
$p_X\c\Pt\to X$ corresponds to the inclusion of variables; see
\Cref{sec:lam} for more details.
\begin{notation}
We denote by $j\c \Pt/\gcat \to \gcat$ the forgetful functor given by
$(X, p_X) \mapsto X$.
\end{notation}
\begin{assumptions}
  We assume that $\C$ has finite limits
  and colimits, and that the functor $\Sigma$ has the form $\Sigma = \Pt + \Sigma'$
  where $\Sigma'$ admits free algebras; hence, in particular, the free monad $\Sigmas$ exists.
\end{assumptions}
%
%
%
\begin{definition}
  \label{def:hog}
  A \emph{$\Pt$-pointed higher-order GSOS law} of $\Sigma$ over
  $B$ is a family of
  morphisms
\begin{align}\label{eq:law}
\rho_{X,Y} \c \Sigma (jX \times B(jX,Y))\to B(jX, \Sigma^\star (jX+Y))
\end{align}
dinatural in $X\in\Pt/\gcat$ and natural in $Y\in \gcat$.
\end{definition}
\noindent Laws of the form \eqref{eq:law-specialcase} emerge from~\eqref{eq:law} 
by choosing $\Pt=0$, the initial object of $\gcat$. When running the semantics, 
both $X$ and $Y$ will be matched with the free algebra~$\mS$
-- abstracting from this choice ensures that the corresponding terms are used in 
a suitably polymorphic, uniform way.


\begin{rem}
  For every $\Sigma $-algebra $(A,a)$, we regard $A$ as a $\Pt$-pointed
  object with point
  \[p_A = \bigl(\Pt\xra{\inl} \Pt+\Sigma' A = \Sigma  A \xra{a} A\bigr).\] Note that
  if $h\colon (A,a)\to (B,b)$ is a morphism of $\Sigma $-algebras,
  then $h$ is also a morphism of the corresponding $\Pt$-pointed
  objects.
\end{rem}

Every object $X \in \gcat$ induces an endofunctor
$B(X,\argument) \c \gcat \to \gcat$. For instance, the transition
system $\gamma_{u}\c\skitermu \to \skitermu + \skitermu^{\skitermu}$
from \eqref{eq:skiutr} is a $B_u(\mS, \argument)$-coalgebra. The state
space~$\skitermu$ is the initial $\Sigma$-algebra for the
corresponding polynomial set functor $\Sigma$; the codomain is
$B_u(\mS,\mS)$. The definition of the map~$\gamma_{u}$ is inductive on
the structure of terms. It turns out to be an instance of a definition
by initiality\sgnote{that is what I would precisely call "by
  induction". shall we explain how syntactic induction corresponds to
  categorical induction in prelims? SM: The text is clear here; I do
  not think we need extra explanations in prelims.}  in which we assign to a
$\Pt$-pointed higher-order GSOS law its canonical operational~model.
For technical reasons, we formulate the abstract definition of
$\gamma_{u}$ yet more generally, by parametrizing it with a
$\Sigma$-algebra $(A,a)$ --- the motivating instance is
obtained by instantiating $A$ with the initial algebra $\mS$.



\begin{lemma}\label{lem:clubs}
Given a $\Pt$-pointed higher-order GSOS law $\rho$ as in \eqref{eq:law}, every $\Sigma $-algebra $(A,a)$ induces a
\emph{unique} morphism $a^\clubsuit\c\mS\to B(A,A)$ in $\C$ such that the following
diagram commutes:
%
%
%
\begin{equation*}
\begin{tikzcd}[column sep=8ex, row sep=normal]
\Sigma(\mS) 
  \dar["\Sigma \brks{\iter a,\,a^\clubsuit}"']
  \ar[rrr,"\iota"] 
  &[-2ex] &[1ex] &[2ex]
\mS 
  \dar["a^\clubsuit"]
  \\
\Sigma (A\times {B(A,A)})
  \rar["{\rho_{A,A}}"] 
&
B(A,\Sigma^\star(A+A))
  \rar["{B(\id,\Sigmas\nabla)}"] &
B(A,\Sigma^\star A)
  \rar["{B(\id, \hat a)}"] &
B(A,A)
\end{tikzcd}
\end{equation*}
Here $\iota \c  \Sigma(\mS)  \to \mS $ is the initial $\Sigma $-algebra,
$\nabla\c A+A\to A$ is the codiagonal, and $\hat a\colon \Sigmas  A \to
A$ is the Eilenberg-Moore algebra
corresponding to $a\colon \Sigma  A\to A$.
\end{lemma}

\begin{rem}\label{rem:a-clubsuit-sigma-morphism}
It follows that $\langle \iter a, a^\clubsuit\rangle$ is a $\Sigma$-algebra morphism from $\mS$ to the $\Sigma$-algebra given by the lower row of the commutative diagram below:
\[
\begin{tikzcd}[column sep=7ex, row sep=normal]
\Sigma (\mS) 
  \dar["\Sigma\brks{\iter a,\,a^\clubsuit}"']
  \ar[rrr,"\iota"] 
  &[1.5ex] &[2ex] &
\mS 
  \dar["\brks{\iter a,\,a^\clubsuit}"]
  \\
\Sigma (A\times {B(A,A)})
  \rar["\brks{a\cdot\Sigma\fst,\,\rho_{A,A}}"] 
&
A\times B(A,\Sigma^\star(A+A))
  \rar["{\id\times B(\id,\Sigmas\nabla)}"] &
A\times B(A,\Sigma^\star A)
  \rar["{\id\times B(\id, \hat a)}"] &
A\times B(A,A)
\end{tikzcd}
\]
\end{rem}

\begin{definition}\label{def:operational-model}
The \emph{operational model} of $\rho$ is given by the $B(\mS,-)$-coalgebra
\[\iota^\clubsuit\c\mu \Sigma \to B(\mS ,\mS ).\]
\end{definition}

\begin{example}\label{ex:a-clubsuit-ski}
Consider the higher-order GSOS law $\rho$ corresponding to $\SKI_u$ (see \Cref{ex:ski-to-ho,ex:ski-to-rho}). Then $\iota^\clubsuit$ is precisely the transition system induced by the rules in \Cref{fig:skirules}. More generally, for a $\Sigma$-algebra $(A,a)$, the morphism $a^\clubsuit$ is obtained by interpreting all those transitions in the algebra $A$. For instance, since there is a transition $K\xto{t}K'(t)$, we have $a^\clubsuit(K)\in A^A$ given by $a^\clubsuit(K)(u)=(K')^A(u)$, where $(K')^A\c A\to A$ is the interpretation of $K'\in \Sigma$ in $A$.
\end{example}

\begin{rem}\label{rem:gsos-to-ho-gsos}
  Every GSOS law $\lambda_Y\c \Sigma(Y\times FY)\to F\Sigmas Y$
  ($Y\in \C$) can be turned into an equivalent~$0$-pointed
  higher-order GSOS law
  $\rho_{X,Y}\c \Sigma(X\times B(X,Y))\to B(X,\Sigmas(X+Y))$
  ($X,Y\in \C$) where
  \[
    B(X,Y)=FY\quad\text{and}\quad \rho_{X,Y} \;=\; \Sigma(X\times FY)  \xto{\Sigma(\inl\times F\inr)} \Sigma((X+Y)\times F(X+Y))  \xto{\lambda_{X+Y}} F\Sigmas(X+Y).
\]
It is not difficult to verify that the operational models of $\lambda$
(see \Cref{sec:abstract-gsos}) and $\rho$ coincide.
\end{rem}

\subsection{Compositionality of Higher-Order Abstract GSOS}
We now investigate when a higher-order GSOS law gives rise to a compositional semantics.
Recall from \Cref{sec:abstract-gsos} that in first-order GSOS, compositionality comes for free and is an immediate consequence of $\mS$ and $\nu B$ extending to initial and final bialgebras, respectively, for a given law. We shall see in \Cref{sec:bialg} that the latter is no longer true in the higher-order setting. Therefore, we take a different route to compositionality, working in a framework of regular categories (\Cref{sec:categories}).

Assuming that the final $B(\mS ,\argument)$-coalgebra
\[
 \zeta \c
  \finalc
  \to
  B(\mS , \finalc)
\]
exists, we think of the unique coalgebra morphism
$\coiter \iota^\clubsuit \c \mS  \to Z$ as the map assigning to each
program in $\mS $ its abstract behaviour. The ensuing notion
of \emph{behavioural equivalence} is then expressed categorically by
the kernel pair of $\coiter\iota^\clubsuit$, i.e.\ the pullback
\begin{equation}
  \label{eq:kernel}
\begin{tikzcd}[column sep=3em, row sep=normal]
  E
  \rar["p_1"]
  \dar["p_2"']
  \ar[dr, phantom , very near start, color=black]
  \pullbackangle{-45}
  &
\mS 
  \dar["\coiter \iota^\clubsuit"]\\
\mS  \rar["\coiter \iota^\clubsuit"] & \finalc
\end{tikzcd}
\end{equation}
We say that this kernel pair forms a \emph{congruence} if the
coequalizer $q\c \mS \epito Q$ of $p_1,p_2$ can be equipped with a
$\Sigma $-algebra structure $a\colon \Sigma  Q \to Q$ such that $q$
is a $\Sigma $-algebra morphism from~$\mS $ to $(Q,a)$, that is,
$q=\iter a$.

\begin{rem}
  For $\C=\Set$ and $\Sigma$ a polynomial functor, the kernel $E$ is
  the equivalence relation on~$\mS $ defined by
  \[
    E = \{\, (s,t)\in \mS \times \mS \c (\coiter\iota^\clubsuit)(s)=(\coiter\iota^\clubsuit)(t)\, \}
  \]
  with the two projection maps $p_1(s,t)=s$ and $p_2(s,t)=t$, and $E$
  forms a congruence in the above categorical sense if and only if it
  forms a congruence in the usual algebraic sense as recalled in \autoref{sec:categories}, i.e.~an equivalence relation
  that is compatible with the $\Sigma $-algebra
  structure of $\mS $.
\end{rem}
\noindent Our main compositionality result asserts that for a regular base category~$\C$ and under mild conditions on the functors $\Sigma$ and $B$, behavioural equivalence is a congruence.
\begin{rem}\label{rem:reflexive}
  Recall that a parallel pair $f,g\c X \parto Y$ is \emph{reflexive}
  if there exists a common splitting, viz.~a morphism $s\c Y \to X$ such that
  $f\comp s = \id_Y = g \comp s$. A \emph{reflexive coequalizer} is a
  coequalizer of a reflexive pair. Preservation of reflexive
  coequalizers is a relatively mild condition for set functors. In
  particular, every polynomial set functor $\Sigma$ and, more generally,
  every finitary set functor preserves reflexive coequalizers~\cite[Cor.~6.30]{AdamekEA11}.
\end{rem}
\begin{theorem}
  \label{th:main}
  Let $\rho$ be a $\Pt$-pointed higher-order GSOS law of
  $\Sigma \c \gcat \to \gcat$ over
  $B\c \gcat^{\opp} \product \gcat \to \gcat$. Suppose that the
  category $\C$ is regular, that $\Sigma $ preserves reflexive
  coequalizers, and that $B$ preserves monomorphisms. Then the kernel
  pair of $\coiter \iota^\clubsuit \c \mS \to \finalc$ is a
  congruence.
\end{theorem}

\noindent (Note that a morphism $(f,g)$ in
$\gcat^{\opp} \product \gcat$ is monic iff $g$ is monic in~$\gcat$ and
$f$ is epic in $\gcat$.) We make use of the following two lemmas,
which are of independent interest, en route to proving
\Cref{th:main}. First, we establish a crucial connection between
$\iota^\clubsuit$ and $a^\clubsuit$ (for general $a\c\Sigma A\to A$),
showing that they can be unified by running the unique morphism
$\iter a\c\mS\to A$ at the covariant and the contravariant positions
of $B$ correspondingly. This critically relies on (di)naturality
of~\eqref{eq:law}.
%
\begin{lemma}\label{lem:law_comm}
Let $(A,a)$ be a $\Sigma $-algebra. Then the following diagram commutes:
%
\begin{equation*}
\begin{tikzcd}[column sep=10em, row sep=normal]
\mS 
  \rar["\iota^\clubsuit"]
  \dar["a^\clubsuit"'] &
B(\mS ,\mS )
  \dar["{B(\id, \iter a)}"] \\
B(A,A)
  \rar["{B(\iter a,\id)}"] &
{B(\mS , A)}
\end{tikzcd}
\end{equation*}
\end{lemma}

\begin{example}[$\SKI_u$, cf.\ \Cref{ex:a-clubsuit-ski}]
The two legs of the diagram send $K\in \mS$ to   $f\in A^\mS$ given by $f(t)=(K')^A((\iter a)(t))$, equivalently $f(t)=(\iter a)(K'(t))$ since $\iter a$ is a $\Sigma$-algebra morphism.
\end{example}

\begin{proof}[Proof of \Cref{lem:law_comm}]
We strengthen the claim a bit and show that the outside of the following diagram commutes:
\begin{equation*}
\begin{tikzcd}[column sep=-3em, row sep=3ex]
	\mS  & &[35ex] &[-4ex] \mS \times B(\mS ,\mS )  \\[3ex]
	& \mS \times B(A,A)  & \mS \times B(\mS , A)  \\[1ex]
	A\times B(A,A) & & & A\times B(\mS , A)
	\arrow[from=2-2, to=2-3, "{\id \times B(\iter a, \id)}"]
	\arrow[from=2-3, to=3-4, "{\iter a\times\id}"']
	\arrow[from=1-1, to=2-2, "\brks{\id,\,a^\clubsuit}" {yshift=-4pt, xshift=6pt}]
	\arrow[bend left=6pt, from=1-1, to=2-3, "\iter b" {yshift=-5pt, xshift=22pt}]
  \arrow[from=1-4, to=2-3, "{\id\times B(\id ,\iter a)}"']
	\arrow[from=1-1, to=1-4, "\brks{\id,\,\iota^\clubsuit}"]
	\arrow[from=2-2, to=3-1, "{\iter a\times\id}"]
	\arrow[from=3-1, to=3-4, "{\id\times B(\iter a,\id)}"]
	\arrow[bend right=20pt, from=1-1, to=3-1, "\brks{\iter a,\,a^\clubsuit}"']
	\arrow[bend left=25pt,  from=1-4, to=3-4, "{\iter a\times B(\id,\iter a)}"]
\end{tikzcd}
\end{equation*}
To this end, we prove commutativity of all inner cells, where $b$ is the $\Sigma $-algebra structure
%
\[
\begin{tikzcd}[column sep=1.9em]
  \Sigma (\mS \times {\BmS A })
  \ar{r}
  [yshift=.5em]
  {\langle \iota\comp \Sigma\fst,\, \rho_{\mS,A}\rangle}
  &
  \mS\times B(\mS,\Sigmas(\mS+A))
  \ar{r}[yshift=.5em]{\id\times B(\mS,\Sigmas[\iter a,\id])}
  &
  \mS\times B(\mS,\Sigmas A)
  \ar{r}
  [yshift=.5em]
  {\id\times B(\id,\hat a)}
  &
  \mS \times\BmS A.
\end{tikzcd}
\]
The bottom quadrangular cell as well as the left and the right triangles obviously commute.
We complete the proof by showing that the compositions
\begin{align}
\mS \xto{\brks{\id ,\,\iota^\clubsuit}}   &\;\mS \times B(\mS,\mS)  \xto{\id\times B(\id, \iter a)} \mS \times B(\mS ,A)\label{eq:law_comm2} \\
\mS \xto{\brks{\id ,\,a^\clubsuit}}       &\;\mS \times B(A,A)      \xto{\id\times B(\iter a, \id)} \mS \times B(\mS ,A),\label{eq:law_comm1}
\end{align}
are both $\Sigma $-algebra morphisms from $\mS$ to $(\mS\times B(\mS,A),b)$, hence equal to $\iter b$ by initiality of $\mS$.

The morphism \eqref{eq:law_comm2} is a composition of $\Sigma $-algebra morphisms:
$\brks{\id,\,\iota^\clubsuit}$ is so by definition, and $\id\times B(\id,\iter a)$
is so by commutativity of the following diagram:
%
%
\begin{equation*}
\begin{tikzcd}[column sep = 8em]
\Sigma (\mS \times B(\mS ,\mS ))
  \ar[dd,swap,"{\brks{\iota\comp\Sigma\fst,\, \rho_{\mS ,\mS }}}"]
  \ar[r,"{\Sigma (\id\times B(\id,\iter a))}"]
&
\Sigma (\mS\times B(\mS, A))
  \ar[d,"{\brks{\iota\comp\Sigma \fst,\,\rho_{\mS ,A}}}"]
\\
&
\mS\times B(\mS, \Sigmas (\mS+A))
  \ar[d,"{\id\times B(\id, \Sigmas(\iter a + \id))}"]
\\
\mS\times B(\mS ,\Sigmas  (\mS +\mS ))
  \ar[d,swap,"{\id\times B(\id,\Sigmas\nabla)}"]
  \ar[r,"{\id\times B(\id, \Sigmas(\iter a+\iter a))}"]
  \ar[ur,near end,bend left=10pt,"{\id\times B(\id,\Sigmas(\id+\iter a))}"]
&
\mS \times B(\mS ,\Sigmas (A+A))  
  \ar[d,"{\id\times B(\id,\Sigmas\nabla)}"]
\\
\mS\times B(\mS,\Sigmas  \mS )
  \ar[d,swap,"{\id\times B(\id,\hat \iota)}"]
  \ar[r,"{\id\times B(\id,\Sigmas(\iter a))}"]
&
\mS\times B(\mS ,\Sigmas  A)
  \ar[d,"{\id\times B(\id,\hat a)}"]
\\
\mS\times B(\mS,\mS )
  \ar[r,"{\id\times B(\id,\iter a)}"]
&
\mS \times B(\mS ,A)
\end{tikzcd}
\end{equation*}
The three upper cells commute by naturality of $\rho$ and by functoriality of $B$
in the second argument; the bottom cell commutes because $\iter a\c {\mS\to A}$ 
is a $\Sigma $-algebra morphism. 

That the morphism~\eqref{eq:law_comm1} is a $\Sigma $-algebra morphism
is shown from the following diagram:
\begin{equation*}
\begin{tikzcd}[column sep=2em, row sep=normal]
  \Sigma(\mS)
  \ar[dddd,"\iota"'] 
  \rar["\Sigma\brks{\id,\, a^\clubsuit}"]
  & 
  \Sigma (\mS \times{B(A,A)})
  \rar["\Sigma (\id\times{B(\iter a,\,\id)})"]
  \dar["\brks{\iota\comp\Sigma\fst,\,\Sigma (\iter a\times\id)}"']
  &[3em]
  \Sigma (\mS\times {\BmS A })
  \dar["\brks{\iota\comp\Sigma\fst,\,\rho_{\mS ,A}}"]
  \\
  & \mS\times \Sigma (A\times B(A,A))
  \dar["\id\times\rho_{A,A}"'] &
  \mS\times\BmS{\Sigma^\star (\mS +A)}
  \dar["{\id\times\BmSf{\Sigma^\star(\iter a+\id)}}"] \\
  &
  \mS \times B(A,\Sigma^\star (A+A))
  \dar["{\id\times B(\id,\Sigma^\star\nabla)}"']
  \rar["{\id\times B(\iter a,\,\id)}"] &
  \mS \times\BmS{\Sigma^\star (A+A)}
  \dar["{\id\times\BmSf{\Sigma^\star\nabla}}"] \\
  &
  \mS\times B(A,\Sigma^\star A)
  \rar["{\id\times B(\iter a,\,\id)}"]
  \dar["{\id\times B(\id,\hat a)}"'] &
  \mS \times \BmS{\Sigma^\star A}
  \dar["{\id\times\BmSf{\hat a}}"] \\
  \mS
  \rar["\brks{\id,\,a^\clubsuit}"]&
  \mS \times B(A,A)
  \rar["{\id\times B(\iter a,\,\id)}"] &
  \mS \times \BmS A
\end{tikzcd}
\end{equation*}
%
The left cell commutes by definition of~$a^\clubsuit$. The two lower right cells 
commute by functoriality of $B$, and the
right upper cell commutes by an instance of dinaturality for $\rho$:
%
\begin{equation*}
\begin{tikzcd}[column sep = 15,row sep=2ex, baseline = (B.base)]
  &
  \Sigma(\mS\times B(\mS ,A))
  \rar["\rho_{\mS ,A}"]
  &[1em]
  B(\mS, \Sigma^\star (\mS +A))
  \ar{dr}{{B(\id,\Sigma^\star(\iter a+\id))}}
  \\
  \Sigma (\mS\times{B(A,A)})
  \ar[ru,"{\Sigma(\id\times{B(\iter a,\id)})}"]
  \ar[rd,"{\Sigma(\iter a\times\id)}"']
  & & &[-1.5ex]
  B(\mS ,\Sigma^\star(A+A))
  \\
  &
  |[alias = B]|
  \Sigma (A\times{B(A,A)})
  \rar["\rho_{A,A}"]
  &
  B(A,\Sigma^\star (A+A))
  \ar[ur,"{B(\iter a,\id)}"']
\end{tikzcd}      \tag*{\qedhere}
\end{equation*}
%
%
%
\end{proof}
\takeout{
\begin{lemma}
For every $\Sigma $-algebra $(A,a)$ the following diagram commutes:
\begin{equation*}
\begin{tikzcd}[column sep=10em, row sep=normal]
\mS 
  \rar["\iter\brks{\id,\,\iota^\clubsuit}"]
  \dar["\iter \brks{\iter a,\,a^\clubsuit}"'] &
\mS \times B(\mS ,\mS )
  \dar["{\iter a\times B(\mS ,\iter a)}"] \\
A\times B(A,A)
  \rar["{A\times B(\iter a,A)}"] &
A\times {B(\mS ,A)}
\end{tikzcd}
\end{equation*}
\end{lemma}
\begin{proof}
  First, recall that $A\times {B(A,A)}$ and
  $\mS \times {B(\mS ,\mS )}$ are $\Sigma $-algebras with the
  structures~$\brks{\iter a,\,a^\clubsuit}$ and $\brks{\id,\,\iota^\clubsuit}$ obtained as
  in~\eqref{eq:iter}. In addition, we define $\Sigma $-algebra
  structures $\alpha$ and $\alpha'$ on $\mS \times B(A,A)$ and
  $\mS \times B(\mS ,A)$, respectively, as follows:
  \[
    \begin{tikzcd}[column sep = 60]
      \Sigma (\mS \times B(A,A))
      \ar{d}{\brks{\iota\comp \Sigma \fst,\, \Sigma (\iter a\times B(A,A))}}
      \ar[shiftarr = {xshift=-60}]{dddd}[swap]{\alpha}
      &
      \Sigma (\mS \times B(\mS ,A))
      \ar{d}[swap]{ \brks{ \iota\comp \Sigma \fst,\,\rho_{\mS ,A} }  }
      \ar[shiftarr = {xshift=70}]{dddd}{\alpha'}
      \\
      \mS  \times \Sigma (A\times B(A,A))
      \ar{d}{\mS \times \rho_{A,A}}
      &
      \mS  \times B(\mS ,\Sigmas (\mS +A))
      \ar{d}[swap]{\mS \times B(\mS ,\Sigmas (\iter a + A))}
      \\
      \mS  \times B(A,\Sigmas  (A+A))
      \ar{d}{\mS  \times B(A,\Sigmas \nabla) }
      &
      \mS  \times B(\mS ,\Sigmas (A+A))
      \ar{d}[swap]{\mS  \times B(\mS ,\Sigmas \nabla) }
      \\
      \mS  \times B(A,\Sigmas  A)
      \ar{d}{\mS  \times B(A,\hat a)}
      &
      \mS  \times B(\mS ,\Sigmas  A)
      \ar{d}[swap]{\mS  \times B(\mS ,\hat a)}
      \\
      \mS  \times B(A,A)
      &
      \mS  \times B(\mS ,A)
    \end{tikzcd}
  \]

  It is our task to prove that the outside of the diagram below commutes: 
  \[
    \begin{tikzcd}[column sep=50]
      \mS 
      \ar[shiftarr = {xshift=-60}]{dd}[swap]{\iter \brks{\iter a,\,a^\clubsuit}}
      \ar{r}{\iter\brks{\id,\,\iota^\clubsuit}}
      \ar{d}[swap]{\iter \alpha}
      \ar{rd}[near end]{\iter \alpha'}
      & 
      \mS \times B(\mS ,\mS )
      \ar[shiftarr = {xshift=70}]{dd}{\iter a\times B(\mS ,\iter a)}
      \ar[dashed]{d}{\mS  \times B(\mS ,\iter a)}
      \\
      \mS \times B(A,A)
      \ar[dashed]{d}[swap]{\iter a \times B(A,A)}
      \ar[dashed]{r}[swap]{\mS  \times B(\iter a, A)}
      &
      \mS \times B(\mS ,A)
      \ar{d}{\iter a \times B(\mS , A)}
      \\
      A\times B(A,A)
      \ar{r}{A\times B(\iter a,A)}
      & 
      A\times B(\mS ,A)
    \end{tikzcd}
  \]
%
%
%
%
  The lower rectangle and the right-hand part obviously commute. To
  see that the remaining three parts commutes, we will prove that the
  three dashed arrows are appropriate $\Sigma $-algebra morphisms; the
  desired parts then commute by the initiality of $\mS $.
  \begin{enumerate}
  \item To show that
    $\iter a\times B(A,A)\c \mS \times B(A,A)\to A\times B(A,A)$ is a
    morphism, we prove that the following diagram commutes:
    \[
      \begin{tikzcd}[column sep=8em]
        \Sigma (\mS \times B(A,A))
        \ar{d}{\brks{\iota\comp \Sigma \fst,\. \Sigma (\iter a\times B(A,A))}}
        \ar{r}{\Sigma (\iter a \times B(A,A))}
        \ar[shiftarr = {xshift=-70}]{dddd}[swap]{\alpha}
        &
        \Sigma (A\times B(A,A))
        \ar{dd}{ \brks{ a\comp \Sigma \fst,\,\rho_{A,A} }  }
        \ar[shiftarr = {xshift=60}]{dddd}{\brks{\iter a,\,a^\clubsuit}}
        \\
        \mS  \times \Sigma (A\times B(A,A))
        \ar{d}[swap]{\mS \times \rho_{A,A}}
        \\
        \mS  \times B(A,\Sigmas  (A+A))
        \ar{d}[swap]{\mS  \times  B(A,\Sigmas \nabla) }
        \ar{r}{\iter a\times B(A,\Sigmas (A+A))}
        &
        A \times B(A,\Sigmas (A+A))
        \ar{d}{A \times  B(A,\Sigmas \nabla) }
        \\
        \mS  \times B(A,\Sigmas  A)
        \ar{d}[swap]{\mS  \times B(A,\hat a)}
        \ar{r}{\iter a\times B(A,\Sigmas  A)}
        &
        A \times B(A,\Sigmas  A) \ar{d}{A \times B(A,\hat a)}
        \\
        \mS  \times B(A,A)
        \ar{r}{\iter a\times B(A,A)}
        &
        A \times B(A,A)
      \end{tikzcd}
    \]
    Indeed, the upper part commutes because $\iter a$ is a morphism
    of $\Sigma $-algebras, and the other parts trivially commute.

  \item Next, we show that
    $\mS \times B(\mS ,\iter a)\c \mS \times B(\mS ,\mS ) \to
    \mS \times B(\mS ,A)$ is a $\Sigma$-algebra morphism, which amounts to
    commutativity of the following diagram:
    \[
      \begin{tikzcd}[column sep = 80]
        \Sigma (\mS \times B(\mS ,\mS ))
        \ar{dd}[swap]{\brks{\iota\comp \Sigma \fst,\, \rho_{\mS ,\mS } }}
        \ar{r}{ \Sigma (\mS \times B(\mS ,\iter a)) }
        \ar[shiftarr = {xshift=-70}]{dddd}[swap]{\brks{\id,\,\iota^\clubsuit}}
        &
        \Sigma (\mS \times B(\mS ,A))
        \ar{d}{ \brks{ \iota\comp \Sigma \fst,\,\rho_{\mS ,A} }  }
        \ar[shiftarr = {xshift=90}]{dddd}{\alpha'}
        \\
        &
        \mS  \times B(\mS ,\Sigmas (\mS +A))
        \ar{d}{\mS \times B(\mS , \Sigmas (\iter a + A))}
        \\
        \mS  \times B(\mS ,\Sigmas  (\mS +\mS ))
        \ar{d}[swap]{\mS  \times B(\mS ,\Sigmas \nabla) }
        \ar{r}[swap]{\mS \times B(\mS ,\Sigmas (\iter a+\iter a))}
        \ar{ur}[near end]{\mS \times B(\mS , \Sigmas (\mS +\iter a)) }
        &
        \mS  \times B(\mS ,\Sigmas (A+A))  \ar{d}{\mS  \times B(\mS ,\Sigmas \nabla) }
        \\
        \mS  \times B(\mS ,\Sigmas  \mS )
        \ar{d}[swap]{\mS  \times B(\mS ,\hat \iota)}
        \ar{r}{ \mS \times B(\mS ,\Sigmas  \iter a) }
        &
        \mS  \times B(\mS ,\Sigmas  A)
        \ar{d}{\mS  \times B(\mS ,\hat a)}
        \\
        \mS  \times B(\mS ,\mS )
        \ar{r}{ \mS \times B(\mS , \iter a)  }
        &
        \mS  \times B(\mS ,A)
      \end{tikzcd}
    \]
    The upper part commutes by naturality of $\rho_{\mS ,-}$ and the
    lower part because $\iter a\c \mS  \to A$ is a $\Sigma $-algebra
    morphism. The two remaining parts commute trivially.

  \item Finally, we prove that 
    $\mS \times B(\iter a,A)\c \mS \times{B(A,A)}\to \mS \times
    B(\mS ,A)$ is a $\Sigma$-algebra morphism. The relevant diagram is
    \[
      \begin{tikzcd}[column sep=8em, row sep=normal]
        \Sigma (\mS \times{B(A,A)})
        \rar["\Sigma (\mS \times{B(\iter a,\,A)})"]
        \ar{d}{\brks{\iota\comp\Sigma \fst,\,\Sigma (\iter a\times  B(A,A))}}
        \ar[shiftarr = {xshift=-70}]{dddd}[swap]{\alpha}
        &
        \Sigma (\mS \times {B(\mS ,\,A)})
        \dar["\brks{\iota\comp\Sigma \fst,\,\rho_{\mS ,A}}"]
        \ar[shiftarr = {xshift=90}]{dddd}{\alpha'}
        \\
        \mS \times \Sigma (A\times B(A,A))
        \dar["\mS \times\rho_{A,A}"']
        &
        \mS \times B(\mS ,\Sigma^\star (\mS +A))
        \dar["{\mS \times B(\mS ,\Sigma^\star(\iter a+A))}"]
        \\
        \mS \times B(A,\Sigma^\star (A+A))
        \dar["{\mS \times B(A,\Sigma^\star\nabla)}"']
        \rar["{\mS \times B(\iter a,\,\Sigma^\star(A+A))}"]
        &
        \mS \times B(\mS ,\Sigma^\star (A+A))
        \dar["{\mS \times B(\mS ,\Sigma^\star\nabla)}"]
        \\
        \mS \times B(A,\Sigma^\star A)
        \rar["{\mS \times B(\iter a,\,\Sigma^\star A)}"]
        \dar["{\mS \times B(A,\hat a)}"']
        &
        \mS \times B(\mS ,\Sigma^\star A)
        \dar["{\mS \times B(\mS ,\hat a)}"]
        \\
        \mS \times B(A,A)
        \rar["{\mS \times B(\iter a,\,A)}"]
        &
        \mS \times B(\mS ,A)
      \end{tikzcd}
    \]
    The two lower parts commute trivially. For the upper rectangle we
    postcompose with the projections of the product in the lower
    right-hand corner: the left-hand component is easily seen to
    commute, and for the right-hand component we use the following
    instance of dinaturality for $\rho_{-,A}$:
    \[
      \begin{tikzcd}[column sep = -20, row sep = 15, baseline = (B.base)]
        &
        \Sigma (\mS \times B(\mS ,A)
        \ar{r}{\rho_{\mS ,A}}
        &[6em]
        B(\mS ,\Sigma^\star (\mS +A))
        \ar{dr}[near end]{B(\mS ,\Sigma^\star (\iter a+A))}
        \\
        \Sigma (\mS \times{B(A,A)})
        \ar{ru}[near start]{\Sigma (\mS \times{B(\iter a,A)})}
        \ar{rd}[near start,swap]{\Sigma (\iter a\times B(A,A))}
        & & &[-2em]
        B(\mS ,\Sigma^\star (A+A))
        \\
        &
        \Sigma (A\times{B(A,A)})
        \ar{r}{\rho_{A,A}}
        &[2em]
        |[alias = B]|
        B(\Sigma^\star (A+A),A)
        \ar{ur}[near end,swap]{B(\iter a,\Sigma^\star (A+A))}
      \end{tikzcd}
      \tag*{\qedhere}
    \]
  \end{enumerate}
%
%
%
\end{proof}
}
Using the universal property of the pullback~\eqref{eq:kernel}, we obtain a morphism
${s\colon \mS \to E}$ such that $p_1 \comp s = \id$ and $p_2 \comp s = \id$.
It follows that ${p_1^\star, p_2^\star\colon \Sigma^\star E \parto \mS}$ is a reflexive pair in $\C$ with common section $\eta_E \comp s$, where~$\eta$ is the unit of the monad~$\Sigma^\star$. By our assumptions, the coequalizer of $p_1^\star$ and $p_2^\star$ is
preserved by the functor~$\Sigma$. Hence, there exists a $\Sigma$-algebra structure
${\iotaq\c\Sigma(\mSq)\to\mSq}$, obtained using the universal property
of the coequalizer $\Sigma(\iter\iotaq)$ from the diagram
\begin{equation}\label{eq:it-i-quot-coeq}
\begin{tikzcd}[column sep=5em, row sep=normal]
\Sigma\Sigma^\star E\dar["\iota"']
  \ar[r,shift right=.5ex,"\Sigma p_2^\star"']
  \ar[r,shift left=.5ex,"\Sigma p_1^\star"]
  &
  \Sigma(\mS)\rar["\Sigma (\iter\iotaq)"]\dar["\iota"]
  &
  \Sigma(\mSq) \ar[dashed]{d}{\iotaq}
  \\
  \Sigma^\star E
  \ar[r,shift right=.5ex,"p_2^\star"']
  \ar[r,shift left=.5ex,"p_1^\star"] &
  \mS\rar["\iter\iotaq"] &  \mSq.
\end{tikzcd}
\end{equation}
Here we already denote the coequalizer of $p_1^\star$ and
$p_2^\star$ by $(\iter\iotaq)$, as commutation of the right-hand side
identifies it as the unique $\Sigma$-algebra morphism induced by $\iotaq$.
\begin{lemma}\label{lem:mSq-coalg}
Under the conditions of \Cref{th:main}, there exists a coalgebra structure
  $\varsigma\c\mSq\to B(\mSq,\mSq)$ making the triangle below commute,
  where $\iotaq^\clubsuit=(\iotaq)^\clubsuit$:
  \[
    \begin{tikzcd}[row sep = 2ex]
      & \mS \ar{dl}[swap]{\iter \iotaq} \ar{dr}{\iotaq^\clubsuit} & \\
      \mSq \ar{rr}{\varsigma} & & B(\mSq,\mSq) 
    \end{tikzcd}
  \]
\end{lemma}
\begin{proof}
By definition of $\iter\iotaq$ as a coequalizer of $p_1^\star$ and $p_2^\star$,
it suffices to show that $\iotaq^\clubsuit$ also coequalizes $p_1^\star$ and $p_2^\star$,
which we strengthen to $\brks{\iter\iotaq,\, \iotaq^\clubsuit}\comp p_1^\star=\brks{\iter\iotaq, \iotaq^\clubsuit}\comp p_2^\star$. Since $\brks{\iter\iotaq,\, \iotaq^\clubsuit}\comp p_1^\star$ and $\brks{\iter\iotaq,\, \iotaq^\clubsuit}\comp p_2^\star$
are $\Sigma$-algebra morphisms (\Cref{rem:a-clubsuit-sigma-morphism}) whose domain is the free $\Sigma$-algebra $\Sigma^\star E$,
it suffices to show that the desired equation holds when precomposed with $\eta_E\colon E \to \Sigma^\star E$.
Thus, it remains to show that $\brks{\iter\iotaq,\, \iotaq^\clubsuit}\comp p_1=\brks{\iter\iotaq,\, \iotaq^\clubsuit}\comp p_2$,
which, in turn, reduces to $\iotaq^\clubsuit\comp p_1 = \iotaq^\clubsuit\comp p_2$.
Next,
\begin{flalign*}
&&& \iotaq^\clubsuit\comp p_1=\iotaq^\clubsuit\comp p_2\\
&&\iff\; & B(\iter\iotaq,\id)\comp \iotaq^\clubsuit\comp p_1= B(\iter\iotaq,\id)\comp \iotaq^\clubsuit\comp p_2&\by{$B(\iter\iotaq,\id)$ is mono}\\
&&\iff\; & \BmSf{\iter\iotaq}\comp \iota^\clubsuit\comp p_1= \BmSf{\iter\iotaq}\comp\iota^\clubsuit\comp p_2.&\by{\autoref{lem:law_comm}}
\end{flalign*}
Let us denote the coequalizer of $p_1, p_2$ by $q\colon \mS \to
Q$. Since~$ \iter\iotaq$ coequalizes $p_1$ and $p_2$, it factorizes
through $q$. It thus suffices to show that
\begin{displaymath}
  \BmSf{q}\comp\iota^\clubsuit\comp p_1
  =
  \BmSf{q}\comp\iota^\clubsuit\comp p_2.
\end{displaymath}
By regularity of the base category~$\gcat$, the unique morphism
$m\colon Q \to Z$ such that $\coit\iota^\clubsuit = m \cdot q$ is monic. Since $\BmSf{\argument}$ 
preserves monomorphisms, it suffices to show that
\begin{displaymath}
  \BmSf{\coit\iota^\clubsuit} \comp\iota^\clubsuit\comp p_1
  =
  \BmSf{\coit\iota^\clubsuit} \comp\iota^\clubsuit\comp p_2.
\end{displaymath}
Note that
$\BmSf{\coit\iota^\clubsuit} \comp\iota^\clubsuit = \zeta\comp
\coit\iota^\clubsuit$ since $\coit\iota^\clubsuit$ is a coalgebra
morphism from $(\mS, \iota^\clubsuit)$ to $(Z, \zeta)$.  Hence the
above equation follows from
$\coit\iota^\clubsuit\comp p_1 = \coit\iota^\clubsuit\comp p_2$, which
holds by~\eqref{eq:kernel}.
\end{proof}

\noindent These preparations in hand, we can proceed with the proof of the main result.%
\begin{proof}[Proof of \Cref{th:main}]
Using the coalgebra $\varsigma\c\mSq\to B(\mSq,\mSq)$ from \autoref{lem:mSq-coalg}
together with \autoref{lem:law_comm}, the upper
rectangular cell of the following diagram commutes:
\begin{equation}\label{eq:cong-argument}
\begin{tikzcd}[column sep=2em, row sep=4ex]
\mS  
  \rar[rr,"\iota^\clubsuit"]
  \dar["\iter\iotaq"']
  \ar[dr, "\iotaq^\clubsuit", near end]
  \ar[shiftarr = {xshift=-30}]{dd}[swap]{\coit\iota^\clubsuit} 
& &[6ex]
B(\mS, \mS)
  \dar["{B(\id, \iter\iotaq)}"]
  \ar[shiftarr = {xshift=45}]{dd}{B(\id,\coit\iota^\clubsuit)}
\\
\mSq
  \ar[r,"\varsigma"]
  \ar{d}[swap]{m}
& B(\mSq,\mSq)
  \rar["{B(\iter\iotaq,\id)}"] 
&
B(\mS,  \mSq)
  \ar{d}{B(\id, m)}
\\
  Z
  \ar[rr,"\zeta"] 
& &
  B(\mS, Z)
\end{tikzcd}
\end{equation}
By finality of $(Z,\zeta)$, we also have a morphism $m$ such
that the lower rectangular cell commutes. Therefore,
$\coit\iota^\clubsuit = m \comp \iter\iotaq$ by uniqueness of
$\coit\iota^\clubsuit$. From this we derive the desired result as
follows. First, we obtain a $\Sigma^\star$-algebra structure
$e\c\Sigma^\star E\to E$ such that $p_1\comp e = p_1^\star$ and
$p_2\comp e = p_2^\star$ by the universal property of~$E$ as the
pullback~\eqref{eq:kernel}, using the fact that
$\coit\iota^\clubsuit\comp p_1^\star = m\comp(\iter\iotaq)\comp
p_1^\star= m\comp(\iter\iotaq)\comp p_2^\star =
\coit\iota^\clubsuit\comp p_2^\star$ by~\eqref{eq:it-i-quot-coeq}.
Since we also have $p_1 = p_1^\star\comp\eta_E$ and
$p_2 = p_2^\star\comp\eta_E$, it follows that the pairs $p_1,p_2$  and $p_1^\star,p_2^\star$ have the
same coequalizer, viz.\ $\iter\iotaq$. It follows that~$p_1,p_2$ is a congruence, since its
coequalizer $\iter\iotaq$ is a $\Sigma$-algebra morphism by
\eqref{eq:it-i-quot-coeq}.
%
%
%
\end{proof}
\begin{example}
  The $\mathrm{SKI_{u}}$ calculus (\Cref{sec:unary}) satisfies the
  assumptions of \Cref{th:main}: $\Set$ is a regular category, every
  polynomial functor $\Sigma$ preserves reflexive coequalizers (see
  \Cref{rem:reflexive}), and the behaviour functor $B_u(X,Y)=Y+Y^X$
  maps surjections to injections in the contravariant argument and
  preserves injections in the covariant one. Consequently,
  compositionality of $\SKI_u$ (\Cref{prop:skicong1}) is
  an instance of \Cref{th:main}. More generally, every $\HO$ specification (see \Cref{def:hoformat}) induces a compositional semantics.
\end{example} 
\begin{example}
  The nondeterministic $\mathrm{SKI_{u}}$ calculus (\Cref{sec:nd-ski})
  is handled analogously; just observe that the finite powerset
  functor $\mypowfin$ preserves both surjections and injections. Thus,
  compositionality (\Cref{prop:skicong2}) follows from \Cref{th:main}.
\end{example}
\noindent A more intricate application of our main theorem is given in \Cref{sec:lam}.

\subsection{Higher-Order Bialgebras}\label{sec:bialg}
We conclude this section with a bialgebraic perspective on higher-order GSOS laws.
\begin{definition}[Higher-Order Bialgebra]
Given a $\Pt$-pointed higher-order GSOS law $\rho$, a \emph{$\rho$-bialgebra} is a triple $(A,a\c\Sigma A\to A,c\c A\to B(A,A))$ such that
the following diagram commutes:
\begin{equation*}
\begin{tikzcd}[column sep=6ex, row sep=normal]
\Sigma A
  \rar["a"]
  \dar["\Sigma\brks{\id,c}"'] &
 A
  \rar["c"] &[5ex]
B( A, A)\\
\Sigma(A\times B(A, A))
  \rar[r,"\rho_{A,A}"]  
&
B(A,\Sigma^\star( A+ A))
  \rar["{B(\id,\Sigmas\nabla)}"]
&
B(A,\Sigma^\star A)
  \uar["{B(\id,\hat a)}"']
\end{tikzcd}
\end{equation*}
A \emph{morphism} from $(A,a,c)$ to another $\rho$-bialgebra $(A',a',c')$ is a morphism $h\c A\to A'$ of $\C$ such that $h\c (A,a)\to (A',a')$ is a $\Sigma$-algebra morphism and the following diagram commutes: 
\begin{equation*}
\begin{tikzcd}[column sep=5ex, row sep=normal]
A
  \ar[rr,"c"]
  \dar["h"'] & &[4ex]
B(A,A)
  \dar["{B(\id,h)}"] \\
A'\ar[r,"c'"] 
& 
B(A',A')
  \rar["{B(h,\id)}"] 
& 
B(A,A')
\end{tikzcd}
\end{equation*}
\end{definition}
\begin{proposition}\label{prop:rho-bialg-cat}
  All $\rho$-bialgebras and their morphisms form a category.%
  \smnote{Sentences should not start with symbols; let alone theorems.}
\end{proposition}

%
As in the first-order case (\Cref{sec:abstract-gsos}), the initial algebra $\mS$ extends to an initial $\rho$-bialgebra:
\begin{proposition}\label{prop:initial-bialgebra}
The triple $(\mS,\iota,\iota^\clubsuit)$ is an initial $\rho$-bialgebra. 
\end{proposition}
In contrast to the first-order case, however, it is generally impossible to derive a final $\rho$-bialgebra from a final $B(\mS,-)$-coalgebra, which is the intended semantic domain for the GSOS law $\rho$:
\begin{example}\label{ex:final-bialgebra}
Consider the bifunctor $B(X,Y)=2^X\c\Set^\opp\times \Set \to \Set$, the empty signature $\Sigma=\emptyset$, and the unique higher-order GSOS law $\rho$ of $\Sigma$ over $B$. A $\rho$-bialgebra is just a map $z\c Z\to 2^Z$, and a morphism from a $\rho$-bialgebra $(W,w)$ to $(Z,z)$ is a map $h\c W\to Z$ making the diagram
\begin{equation}\label{eq:bialg-morph-ex}  
\begin{tikzcd}
W\ar{d}[swap]{h} \ar{rr}{w} & & 2^W \ar[equals]{d}\\
Z \ar{r}{z} & 2^Z \ar{r}{2^h} & 2^W
\end{tikzcd}
\end{equation}
commute. We claim that no final $\rho$-bialgebra exists, despite the endofunctor $B(\mS,-)=2^0\cong 1$ having a final coalgebra. Suppose for a contradiction that $(Z,z)$ is a final $\rho$-bialgebra. Choose an arbitrary $\rho$-bialgebra $(W,w)$ such that $\under{W}>\under{Z}$ and $w\c W\to 2^W$ is injective. Then no map $h\c W\to Z$ makes \eqref{eq:bialg-morph-ex} commute, since $w$ is injective but $h$ is not. 
\end{example}
On the positive side, we have an algebra structure
$\iotaq\c\Sigma(\mSq)\to\mSq$ by \Cref{th:main} and a
coalgebra structure $\varsigma\c \mSq\to B(\mSq,\mSq)$ by \Cref{lem:mSq-coalg}, and these combine to a $\rho$-bialgebra:\lsnote{@Sergey: corrected this, please check}
\begin{proposition}\label{prop:sim-bialg}
  Under the conditions of \Cref{th:main}, the triple $(\mSq,\iotaq,\varsigma)$ is a $\rho$-bialgebra.
\end{proposition}
The above results indicate that bialgebras remain a meaningful concept in higher-order GSOS, but in comparison to first-order GSOS are less useful as a tool for deriving congruence results.

\section[The Lambda-calculus]{The \texorpdfstring{$\boldsymbol{\lambda}$}{$\lambda$}-calculus}
\label{sec:lam}

We now depart from combinatory calculus and move to languages with
variable binding, starting with the all-important (untyped) $\lambda$-calculus.
The $\lambda$-calculus comes in various flavours, such as
\emph{call-by-name} or \emph{call-by-value}, and the respective
operational semantics can be formulated in either big-step or
small-step style. For the purposes of our work, we are going to give a
categorical treatment of the small-step call-by-name and the small-step
call-by-value $\lambda$-calculus. We start with the former, whose operational semantics is
presented in \Cref{fig:lambda}. Here, $p,p',q$ range over possibly open $\lambda$-terms and $[q/x]$ denotes capture-avoiding substitution of the term $q$ for the variable $x$.
%
\begin{figure}[h]
  \[
    \begin{array}{l@{\qquad}l@{\qquad}l}
      \inference[\texttt{app1}]{\goes{p}{p\pr}} {\goes{p \app q}{p\pr \app q}}
      &
        \inference[\texttt{app2}]{}{\goes{(\lambda x.p) \app q}{p[q/x]}}
    \end{array}
  \]
  \caption{Small-step operational semantics of the call-by-name $\lambda$-calculus.}
  \label{fig:lambda}
\end{figure}
%
%

\noindent The operational semantics of the call-by-name
$\lambda$-calculus induces a deterministic transition relation $\to$ on the set of $\lambda$-terms modulo $\alpha$-equivalence. Every $\lambda$-term $t$ either \emph{reduces} ($t\to t'$ for some~$t'$) or is in \emph{weak head normal form}, that is, $t$ is a $\lambda$-abstraction $\lambda x.t'$ or of the form $x \app s_{1} \app s_{2} \app \cdots \app
  s_{k}$ for a variable $x$ and terms $s_1,\ldots,s_k$ ($k\geq 0$). As usual, we let application associate to the left: $t_1\app t_2\app t_3 \app \cdots \app t_n$ means $(\cdots((t_1\app t_2)\app t_3) \cdots) \app t_n$.


On the side of program equivalences, $\lambda$-calculus semantics can
be roughly divided into three kinds: \emph{applicative
  bisimilarity}~\cite{Abramsky:lazylambda}, \emph{normal form
  bisimilarity}~\cite{DBLP:conf/lics/Lassen05} and \emph{environmental
  bisimilarity}~\cite{DBLP:conf/lics/SangiorgiKS07}. We are looking to
give a coalgebraic account of \emph{strong} versions of applicative bisimilarity, see \Cref{def:strong-app} and \Cref{prop:bisim-vs-appbisim}.

\subsection{The Presheaf Approach to Higher-Order Languages}
\label{sec:prelims}

\citet{DBLP:conf/lics/FiorePT99} propose the
presheaf category $\vcat$ as a setting for algebraic signatures with
variable binding, such as the $\lambda$-calculus and the
$\pi$-calculus. We review some of the core ideas from their work as
well as follow-up work by \citet{DBLP:conf/lics/FioreT01}.

Let $\fset$ be the category of finite cardinals, the skeleton of the
category of finite sets. Objects in $\fset$ are thus
sets $n=\{0,\dots,n-1\}\;(n \in \Nat)$, and morphisms $n\to m$ are
functions. The category $\fset$ has a canonical coproduct structure
\begin{equation}
  \label{coproduct}
  n \xrightarrow{\oname{old}_{n}} n + 1 \xleftarrow{\oname{new}_{n}} 1
\end{equation}
where $\oname{old}_{n}(i)=i$ and $\oname{new}_{n}(0)=n$.
Notice the appropriate naming of the coproduct injections: The idea is
that each object $n \in \fset$ is an untyped context of $n$ free
variables, while morphisms $n \to m$ are variable
\emph{renamings}. When extending a context along
$\oname{old}_{n}(i)=i$, we understand the pre-existing elements
of~$n$ as the ``old'' variables, and the added element
$\oname{new}_{n}(0)$ as the ``new'' variable. The coproduct structure
of $\fset$ gives rise to three fundamental operations on contexts,
\emph{exchanging}, \emph{weakening} and \emph{contraction}:
\begin{equation}
  \label{eq:ops}
  \begin{aligned}
    \oname{s} &= [\oname{new}_{1},\oname{old}_{1}] \c 2 \to 2, \\
    \oname{w} &= \oname{old}_{0} \c 0 \to 1, \\
    \oname{c} &= [\id_{1},\id_{1}] \c 2 \to 1.
  \end{aligned}
\end{equation}
We think of a presheaf $X\in \vcat$ as a collection of terms: elements of $X(n)$ are ``$X$-terms'' with
free variables from the set $n=\{0,\ldots,n-1\}$, and for each $r\c n\to m$ the map $X(r)\c X(n)\to X(m)$ sends a term $t\in X(n)$ to the term $X(r)(t)\in X(m)$ obtained by renaming the free variables of $t$ according to $r$.

\begin{example}\label{ex:presheaves}
\begin{enumerate}
\item\label{ex:presheaves-v} The simplest example is the presheaf
  $V\in \vcat$ of variables, defined by
  \begin{equation*}
    V(n) = n \qquad\text{and}\qquad V(r) = r.
  \end{equation*}
  Thus, a $V$-term at stage $n$ is simply a choice of a variable $i \in
  n$.

\item\label{ex:presheaves-sigma} For every algebraic signature
  $\Sigma$, the presheaf $\Sigmas\in \vcat$ of $\Sigma$-terms is given
  by the domain restriction of the free monad on $\Sigma$ to
  $\fset$. Thus $\Sigmas(n)$ is the set of $\Sigma$-terms in variables
  from $n$.
\item\label{ex:presheaves-lambda} The presheaf $\Lambda\in \vcat$ of
  $\lambda$-terms is given by%
  \smnote{In the second line I find the substitution confusing without
    the round inner brackets because $/$ seems to bind stronger than
    forming lists.} 
  \begin{align*}
    \Lambda(n) &= \text{$\lambda$-terms modulo
      $\alpha$-equivalence with free variables from $n$, and}
    \\
    \Lambda(r)(t) & = t[(r(0),\ldots,r(n-1))/(0,\ldots,n-1)]\qquad
    \text{for $r\colon n\to m$}.
  \end{align*}
\end{enumerate}
\end{example}
The idea of substituting terms for variables can be treated at the abstract level of presheaves as follows. For every presheaf $Y \in \vcat$, there is a functor
\begin{equation}\label{eq:tensor}
  \textstyle
  - \mathbin{\mon} Y \c \vcat \to \vcat, \qquad (X \mathbin{\mon} Y)(m) = \int^{n \in \fset} X(n) \product (Y(m))^{n} = \bigl(\coprod_{n\in \fset} X(n)\times (Y(m))^n\bigr)/\approx, 
\end{equation}
where $\approx$ is the equivalence relation generated by all pairs
\[
  (x, y_0,\ldots, y_{n-1}) \approx (x', y_0',\ldots, y_{k-1}')
\]
such that $(x, y_0,\ldots, y_{n-1})\in X(n)\times Y(m)^n$, $(x', y_0',\ldots, y_{k-1}')\in X(k)\times Y(m)^k$ and there
exists $r\c n\to k$ satisfying $x'=X(r)(x)$ and $y_{i}=y'_{r(i)}$ for
$i=0,\ldots, n-1$.  An equivalence class in $(X \mathbin{\mon} Y)(m)$
can be thought of as a term $x\in X(n)$ with $n$ free variables,
together with~$n$ terms $y_0,\ldots,y_{n-1}\in Y(m)$ to be substituted
for them. The above equivalence relation then says that the outcome of the substitution should be invariant under renamings that reflect equalities among $y_1,\ldots,y_n$; for instance, if $y_i=y_j$ and $r\c n\to n$ is the bijective renaming that swaps $i$ and $j$, then substituting $y_1,\ldots, y_n$ for $1,\ldots,n$ in the term $X(r)(x)$ should produce the same outcome. Varying~$Y$,
one obtains the \emph{substitution tensor}
$\argument \mon \argument \c \vcat \product \vcat \to
\vcat$, 
which makes $\vcat$ into a (non-symmetric) monoidal category with
unit~$V$, the presheaf of variables.  Monoids in $(\vcat, \mon, V)$ can be seen as collections of terms equipped with a
substitution structure. 

For given~$Y$, the functor
$- \mathbin{\mon} Y \c \vcat \to \vcat$ has a right adjoint\footnote[1]{\citet{DBLP:conf/lics/FiorePT99} denote the right adjoint by $\langle Y,-\rangle$; we use $\llangle Y,-\rrangle$ instead to distinguish from morphisms into products.}
\begin{equation*}
  \llangle Y, \argument \rrangle \c \vcat \to \vcat, \qquad \llangle Y,W \rrangle(n) = \int_{m \in \fset} [(Y(m))^{n}, W(m)] = \NT(Y^{n} , W).
\end{equation*}
That is, an element of $\llangle Y,W \rrangle(n)$ is a natural family
of maps $Y(m)^n\to W(m)$, to be thought of as describing the substitution of $Y$-terms in $m$ variables for the $n$ variables of a fixed term, resulting in  a $W$-term in $m$ variables. Thus, $(\vcat, \mon, V)$ is in fact a closed monoidal category.

\subsection{Syntax}


%
Variable binding is captured by the \emph{context extension} endofunctor \[\delta \c \vcat \to \vcat\]
defined on objects by
\begin{equation*}
  \delta X (n) = X(n + 1) \qquad\text{and}\qquad \delta X (h) = X(h + \id_{1})
\end{equation*}
and on morphisms $h \c X \to Y$ by
\begin{equation*}
  (\delta h)_{n} =  \big(X(n + 1) \xra{h_{n + 1}} Y(n + 1)\big).
\end{equation*}
Informally, the elements of $\delta X(n)$ arise by binding the last variable in
an $X$-term with $n+1$ free variables.
The operations $\oname{s},\oname{w},\oname{c}$ on contexts, see \eqref{eq:ops},
give rise to natural transformations
\[
  \oname{swap} \c \delta^{2} \to \delta^{2},
  \qquad
  \oname{up} \c \Id \to \delta \qquad\text{and}\qquad
  \oname{contract} \c \delta^{2} \to \delta
\]
in $\vcat$, which correspond respectively to
the actions of swapping the two ``newest'' variables in a term, weakening and
contraction. Their components are defined by 
\begin{equation}
  \label{eq:ops}
  \begin{aligned}
    \oname{swap}_{X,n}
    &=  \big(\,X(n + 2) \xra{X(\id_{n} + \oname{s})} X(n + 2)\,\big), \\
    \oname{up}_{X,n}
    &= \big(\,X(n) \xra{X(\id_{n} + \oname{w})} X(n + 1)\,\big), \\
    \oname{contract}_{X,n}
    &=  \big(\,X(n + 2) \xra{X(\id_{n} + \oname{c})} X(n + 1)\,\big).
  \end{aligned}
\end{equation}

The presheaf $V$ of variables
(\Cref{ex:presheaves}\ref{ex:presheaves-v}) and the endofunctor
$\delta$ are the two main constructs that enable the categorical
modelling of syntax with variable binding. For example, the binding
signature of the $\lambda$-calculus corresponds to the endofunctor
\begin{equation}\label{eq:syn}
  \Sigma \c \vcat \to \vcat,\qquad \Sigma X = V + \delta X + X \product X.
\end{equation}
This is analogous to algebraic signatures determining (polynomial)
endofunctors on $\set$. For $\Sigma$ as in~\eqref{eq:syn}, the
forgetful functor $\alg{\Sigma} \to \vcat$ has a left adjoint that
takes a presheaf $X \in \vcat$ to the free algebra $\Sigma$-algebra
$\Sigmas X$. In particular, the initial algebra $\mS$ is the presheaf
$\Lambda$ of $\lambda$-terms
(\Cref{ex:presheaves}\ref{ex:presheaves-lambda}), a key observation
which is an instance of \cite[Thm.~2.1]{DBLP:conf/lics/FiorePT99}.


\removeThmBraces
\begin{proposition}[\cite{DBLP:conf/lics/FiorePT99}]
  \label{prop:sub}
  The presheaf $\Lambda = \mS$ of
  $\lambda$-terms admits the structure of a monoid $(\Lambda,
  \mu , \eta)$ in $(\vcat, \mon, V)$ whose unit $\eta\c V\to \Lambda$ is the inclusion of variables and whose multiplication $\mu\colon \Lambda\bullet \Lambda \to \Lambda$ is the uncurried form of the natural transformation $\bar \mu\c  \Lambda\to \llangle \Lambda,\Lambda \rrangle$ given by 
\[ \bar\mu_n \c \Lambda(n)\to \llangle \Lambda,\Lambda \rrangle(n)=\NT(\Lambda^n,\Lambda),\qquad  t\mapsto \lambda \vec{u}\in \Lambda(m)^n. \,t[\vec{u}]. \]
Here, $t[\vec{u}]$ denotes the simultaneous substitution $t[(u_0,\ldots, u_{n-1}) / (0, \ldots, {n-1})]$. 
\end{proposition}
\resetCurThmBraces

\subsection{Behaviour}
To capture the $\lambda$-calculus in the abstract categorical setting of higher-order GSOS laws developed in \Cref{sec:hogsos}, we consider the behaviour bifunctor 
\begin{equation}
  \label{def:beh}
  B \c (\vcat)^{\opp} \product \vcat \to \vcat,\qquad B(X,Y) = \llangle X,Y \rrangle \product (Y + Y^{X} + 1),
\end{equation}
where $Y^X$ denotes the exponential object in the topos $\vcat$.
Our intended operational model is a $B(\Lambda,-)$-coalgebra structure
\begin{equation}\label{eq:lambda-op-model}
\langle \gamma_{1},\gamma_{2}\rangle \c \Lambda \to
B(\Lambda,\Lambda)
\end{equation}
on the presheaf of $\lambda$-terms. For each term $t\in \Lambda(n)$, the natural transformation $\gamma_1(t)\c \Lambda^n\to \Lambda$ exposes the simultaneous substitution structure, that is, $\gamma_1(t)$ is equal to $\bar \mu(t)$ from \Cref{prop:sub}. Similarly,
 $\gamma_{2}(t)$ is an element of the coproduct
$\Lambda(n) + \Lambda^{\Lambda}(n) + 1$, representing either a reduction step, a
$\lambda$-abstraction seen as a function on terms, or that $t$ is stuck. To apply the higher-order GSOS framework, let us first note that one of its key assumptions holds:
\begin{lemma}\label{lem:final-coalgebra}
  For every $X\in \vcat$ the functor $B(X,-)$ has a final coalgebra.
\end{lemma}

Since $B(X,-)$ preserves pullbacks, behavioural equivalence on
$B(X,-)$-coalgebras coincides with coalgebraic
bisimilarity~\cite{DBLP:journals/tcs/Rutten00}. Recall that a
\emph{bisimulation} between $B(X,-)$-coalgebras $W \to B(X,W)$ and
$Z \to B(X,Z)$ is a presheaf $R\seq W\times Z$ that can be equipped
with a coalgebra structure $R\to B(X,R)$ such that the two projection
maps $R\to W$ and $R\to Z$ are $B(X,-)$-coalgebra
morphisms. The following proposition gives an elementary
characterization of bisimulations.


\begin{notation}\label{N:exp}
  Recall~\cite[Sec.~I.6]{MacLaneMoerdijk92} that the exponential $Y^X$ in $\vcat$
  and its evaluation morphism $\ev\c Y^X \times X \to Y$ are, respectively, given by
  \[
    Y^X(n) = \NT((\argument)^n \times X, Y)
    \qquad
    \text{and}
    \qquad
    \ev_n(f, x) = f_n(\id_n, x) \in Y(n)
  \]
  for a natural transformation $f\c (\argument)^n \times X \to Y$ and an element $x
  \in X(n)$. In the following we put \[f(x) \,:=\,\ev_n(f,x).\]
\end{notation}
\begin{proposition}
  \label{prop:bisim}
  Given $X\in \vcat$ and two $B(X,\argument)$-coalgebras 
\[\langle c_{1},c_{2}\rangle \c W \to \llangle X,W\rrangle \times (W+W^X+1)\qquad\text{and}\qquad \langle d_{1},d_{2}\rangle \c Z \to \llangle X,Z\rrangle\times(Z+Z^X+1),\] a family of relations
  $R(n)\seq W(n)\times Z(n)$, $n\in \fset$, is a bisimulation if and
  only if for all $n\in \fset$ and $w \mathbin{R(n)} z$ the following hold (omitting 
  subscripts of components of the natural transformations
  $c_i$, $d_i$):
  \begin{enumerate}
  \item $W(r)(w) \mathbin{R(m)} Z(r)(z)$ for all $r \c n \to m$;
  \item\label{bisim:2} $c_{1}(w)(\vec{u})
    \mathbin{R(m)} d_{1}(z)(\vec{u})$ for all $m \in \fset$ and $\vec{u} \in X(m)^{n}$;
  \item\label{bisim:3} $c_{2}(w) =: w\pr \in W(n) \implies d_{2}(z) =: z\pr\in Z(n) \wedge w\pr \mathbin{R(n)} z\pr$;
  \item $c_{2}(w) =: f \in W^{X}(n) \implies d_{2}(z) =: g\in Z^X(n) \wedge \forall e \in
    X(n).\,f(e) \mathbin{R(n)} g(e)$;
  \item $c_{2}(w) = \ast \implies d_{2}(z) = \ast$;
  \item $d_{2}(z) =: z\pr \in Z(n) \implies c_{2}(w) =: w\pr\in W(n) \wedge w\pr \mathbin{R(n)} z\pr$;
  \item $d_{2}(z) =: g \in Z^{X}(n) \implies c_{2}(w) =: f\in W^X(n) \wedge
    \forall e \in X(n).~f(e) \mathbin{R(n)} g(e)$;
  \item \label{bisim:8} $d_{2}(z) = \ast \implies c_{2}(w) = \ast$.
  \end{enumerate}
\end{proposition}
\noindent Condition (1) states that $B(X,-)$-bisimulations are compatible with
the renaming of free variables: given a renaming $r \c n \to m$, the renamed
terms $W(r)(w)$ and $Z(r)(z)$ are related by $R(m)$. Similarly, condition (2)
states that $B(X,-)$-bisimulations are compatible with substitutions: given a
substitution $\vec{u} \in X(m)^{n}$, the resulting terms $c_{1}(w)(\vec{u})$ and
$d_{1}(z)(\vec{u})$ are related by $R(m)$. Conditions (6)--(8) are symmetric to
(3)--(5); in fact, since $B(X,-)$-coalgebras are
deterministic transition systems, the former conditions imply the latter for
every bisimulation $R$. We opted to state (6)--(8) explicitly, as
these conditions become relevant in nondeterministic extensions of the
$\lambda$-calculus.

\subsection{Semantics}


As explained above, in our intended operational model
$\langle \gamma_{1},\gamma_{2}\rangle \c \Lambda \to B(\Lambda,\Lambda)$ the component $\gamma_1$ should be the transpose of the monoid multiplication
$\mu \c \Lambda \bullet \Lambda \to \Lambda$ from \Cref{prop:sub} under
the adjunction
$\argument \mon \Lambda \dashv \llangle \Lambda , \argument
\rrangle$. As an interesting technical subtlety, for this model to be induced by a
higher-order GSOS law $\rho_{X,Y}$ of some sort, the argument $X$ is required to be 
equipped with a \emph{point} $\var \c V \to X$. The importance
of points for defining substitution was first identified by
\cite{DBLP:conf/lics/FiorePT99} (see also \cite{DBLP:conf/lics/Fiore08}) and is
worth recalling from its original source.

Fiore et al.~argued that, given an endofunctor $F \c \vcat \to \vcat$,
in order to define a substitution structure
$F^{\star}V \bullet F^{\star}V \to F^{\star}V$ on the free $F$-algebra
over $V$, it is necessary for $F$ to be \emph{tensorially strong}, in
that there is a natural transformation
$\strength_{X,Y} \c FX \bullet Y \to F(X \bullet Y)$ satisfying the expected
coherence laws~\cite[§ 3]{DBLP:conf/lics/FiorePT99}. For the
special case of~$F$ being the context extension endofunctor $\delta$,
this
 requires the presheaf $Y$
to be equipped with a \emph{point} $\var \c V \to Y$: the strength map
$\strength_{X,Y} \c \delta X \bullet Y \to \delta(X \bullet Y)$ is
given at $m\in\fset$ by%
\smnote{I changed $n$ to $m$ here; otherwise the typing makes no
  sense whatsoever.}
\[
  [t \in X(n+1),\vec{u} \in Y(m)^{n}] \quad\xmapsto{\strength_{X,Y}}\quad
  [t,(\oname{up}_{Y,m}(\vec{u}),\var_{m+1}(\oname{new}_{m}))\in Y(m+1)^{n+1}];
\]
see
\eqref{eq:tensor} for the definition of $\bullet$. Intuitively, given
a substitution of length $n$ on a term with $n+1$ free variables, a
fresh variable in $Y$ should be used to (sensibly) produce a
substitution of length $n+1$.  This situation is relevant in the
context of higher-order GSOS laws of binding signatures over $B(X,Y)$
where, e.g. in the case of the $\lambda$-calculus, one is asked to
define a map of the form (factoring out the unnecessary parts)
\begin{equation}
  \label{eq:rhodelta}
  \rho_{1} \c \delta\llangle X,Y \rrangle \to \llangle X,\delta Y \rrangle.
\end{equation}
Writing $\overline{\strength}$ for the transpose of $\strength$ under $\argument
\mon X \dashv \llangle X , \argument \rrangle$, we obtain $\rho_{1}$ simply as
\begin{equation*}
  \begin{tikzcd}
    \delta\llangle X,Y \rrangle
    \arrow[rr,"\overline{\strength}_{\llangle X,Y \rrangle,X}"]
    & & \llangle X, \delta(\llangle X,Y \rrangle \bullet X)\rrangle
    \arrow[rr,"\llangle X {,} \delta (\varepsilon) \rrangle"]
    & & \llangle X, \delta Y \rrangle,
  \end{tikzcd}
\end{equation*}
where $\varepsilon\c \llangle X,Y \rrangle \bullet X\to Y$ is the evaluation morphism for the hom-object  $\llangle X,
Y\rrangle$. The map $\rho_{1}$ should be considered a form of \emph{capture-avoiding
  substitution}, as in elementary terms it takes a natural transformation $f\colon X^{n+1}\to Y$ to the natural transformation $\rho_1(f)\c X^n\to \delta Y$ given by
\[
  \vec{u}
  \in X(m)^{n} \quad \mapsto \quad f_{m+1}(\oname{up}_{X,m}(\vec{u}),\var_{m+1}(\oname{new}_m)) \in Y(m+1).
\]
Thus $\rho_{1}$ represents simultaneous substitution in which the freshest
variable is bound, hence it should not be substituted. 
At the same time, a higher-order GSOS law for the $\lambda$-calculus
needs to turn a $\lambda$-abstraction into a function on
potentially open terms precisely by only substituting the bound
variable. This implies that we need natural transformation of the form
\begin{equation}
  \label{eq:rholam2}
  \rho_{2} \c \delta\llangle X,Y \rrangle \to Y^{X}.
\end{equation}
Again, we make use of the point $\var\c V \to X$ to produce $\rho_{2}$:
\begin{equation*}
  \begin{tikzcd}
    \delta\llangle X,Y \rrangle
    \arrow[r,"\cong"]
    & \llangle X,Y^{X} \rrangle
    \arrow[rr,"\llangle \var {,} Y^{X} \rrangle"]
    & & \llangle V, Y^{X}\rrangle
    \arrow[r,"\cong"]
    & Y^{X}.
  \end{tikzcd}
\end{equation*}
Here, the first isomorphism is given by
\[ \delta\llangle X,Y\rrangle(n) = \NT(X^{n+1},Y) \cong \NT(X^n,Y^X) = \llangle
  X,Y^X\rrangle (n). \]
Thus, in elementary terms,
\[ \rho_2(f)(e) = f_{n}(\var_n(0),\ldots, \var_n(n-1),e)\qquad\text{for $f\c X^{n+1}\to Y$ and $e\in X(n)$.} \]

With these preparations at hand, we are now ready to define the small-step operational semantics of the
call-by-name $\lambda$-calculus in terms of a
$V$-pointed higher-order GSOS law of the syntax endofunctor $\Sigma X = V+\delta X + X \product X$ over the behaviour bifunctor
$B(X,Y)=\llangle X,Y\rrangle \times (Y+Y^X+1)$. A law of this type is given by a family of presheaf maps
\[
\begin{tikzcd}
  V+\delta(X\times \llangle X,Y\rrangle \product (Y + Y^{X} + 1) ) + (X\times \llangle X,Y \rrangle \product (Y + Y^{X} + 1))^2 \ar{d}{\rho_{X,Y}}  \\
 \llangle X,\Sigmas(X+Y)  \rrangle \product (\Sigmas(X+Y) + (\Sigmas(X+Y))^{X} + 1)
\end{tikzcd}
\]
dinatural in $(X,\var_X)\in V/\vcat$ and natural in $Y\in \vcat$. We let $\rho_{X,Y,n}$ denote the component of $\rho_{X,Y}$ at $n\in \fset$.
\begin{notation}
  We write
  \[\lambda.(-)\c \delta\Sigmas\to \Sigmas \qquad\text{and}\qquad
    \circ\c \Sigmas\times \Sigmas\to \Sigmas \] for the natural
  transformations whose components come from the $\Sigma$-algebra structure
  on free $\Sigma$-algebras; here $\circ$ denotes application. In the following we will consider free
  algebras of the form $\Sigmas(X+Y)$. For simplicity, we usually keep inclusion maps implicit: Given $t_1,t_2\in X(n)$ and
  $t_1'\in Y(n)$ we write $t_1\app t_2$
  for $[\eta\comp \inl(t_1)]\circ [\eta\comp \inl(t_2)]$, and similarly
  $t_1\app t_1'$ for $[\eta\comp\inl(t_1)]\circ [\eta\comp \inr(t_1')]$
  etc., where $\inl$ and $\inr$ are the coproduct injections and $\eta\c \Id\to\Sigmas$ is the unit of the free monad $\Sigmas$.
\end{notation}
\begin{notation}
  Let \[\pi\colon V\to \llangle X,\Sigmas(X+Y)\rrangle \] be the adjoint
  transpose of
  \[
    V\bullet X \xto{\cong} X \xto{\inl} X+Y \xto{\eta}
    \Sigmas(X+Y).
  \]
  Thus for $v\in V(n)=n$, the natural transformation
  $\pi(v)(n)\colon X^n\to \Sigmas(X+Y)$ is the $v$-th projection
  $X^n\to X$ followed by $\eta\comp \inl$. Further, recall 
  that $j \c \Pt/\vcat \to \vcat$ denotes the forgetful functor. 
\end{notation}

\begin{definition}[$V$-pointed higher-order GSOS law for the call-by-name
  $\lambda$-calculus]
  \label{def:lamgsos}
  \begin{align*}
    \rho^{\cn}_{X,Y} \c \quad & \Sigma(jX \times B(jX,Y)) & \to \quad & B(jX, \Sigma^\star (jX+Y)) \\
  \rho^{\cn}_{X,Y,n}(tr) = \quad & \texttt{case}~tr~\texttt{of} \\
  & v \in V(n) & \mapsto \quad & \pi(v),* & \\
  & \mathsf{\lambda}.(t, f,\_) & \mapsto
    \quad & \llangle X, \lambda.(-)\comp \eta \comp \inr \rrangle (\rho_{1}(f)),
   (\eta \comp \inr)^{X}(\rho_{2}(f)) & \\
  & (t_{1}, g, t_{1}\pr) \app (t_{2}, h,\_) & \mapsto
    \quad & \lambda \vec{u}. ( g_{m}(\vec{u}) \app h_{m}(\vec{u}) ),t_{1}\pr \app t_{2} & \\
  & (t_{1}, g, k) \app (t_{2}, h,\_) & \mapsto
    \quad & \lambda \vec{u}. (g_{m}(\vec{u}) \app h_{m}(\vec{u})),\eta \comp \inr \comp k(t_{2}) & \\
  & (t_{1}, g, *) \app (t_{2}, h,\_) & \mapsto
    \quad & \lambda \vec{u}. (g_{m}(\vec{u}) \app h_{m}(\vec{u})),* &
\end{align*}
where $t\in \delta X(n)$, $f\in \delta\llangle X,Y\rrangle(n)$, $g,h\in \llangle X,Y\rrangle(n)$, $\vec{u} \in
X(m)^{n}$ for $m \in \mathbb{N}$, $k\in Y^X(n)$, $t_1, t_2\in X(n)$
and $t_1'\in Y(n)$ (we have omitted the brackets around the pairs on
the right).
\end{definition}

In \Cref{fig:cn} the definition of $\rho^\cn$ is rephrased in the
familiar style of inference rules. Here, the notation $[\ldots]$ and $\to$
refers to the first and second component of $B(jX,Y)$, respectively, and
labelled arrows correspond to function application. For instance, the rule \texttt{lam} expresses that for every $t\in \delta X(n)$, $f\in \delta\llangle X,Y\rrangle$ and $e\in X(n)$, putting $\vec{u} = (\var_n(0),\ldots, \var_n(n-1),e)$ and $t'=f(\vec{u})$, the second component of $\rho_{X,Y,n}(\lambda.(t,f,\_))$ lies in $Y^X(n)$ and satisfies $\rho_{X,Y,n}(\lambda.(t,f,\_))(e)=t'$. This matches precisely the corresponding clause of \Cref{def:lamgsos}.

 \begin{figure}[h]
   \[
     \begin{array}{l@{\qquad}l}
       \inference[\texttt{var}]{}{v \not\xrightarrow{}}
       & \inference[\texttt{lam}]{
       \vec{u} = (\var_n(0),\ldots, \var_n(n-1),e\in X(n))
       & t[\vec{u}] = t'}{\goesv{\lambda.t}{t\pr}{e}}
     \end{array}
   \]
   \[
     \begin{array}{l@{\qquad}l@{\qquad}l}
       \inference[\texttt{app1}]{\goes{t_{1}}{t_{1}\pr}}{\goes{t_{1} \app t_{2}}{t_{1}\pr \app t_{2}}}
       & \inference[\texttt{app2}]{t_{1} \xrightarrow{t_{2}} t_{1}\pr}{\goes{t_{1} \app t_{2}}{t_{1}\pr}}
       &  \inference[\texttt{app3}]{t_{1}\not\xrightarrow{}}{t_{1}\app t_{2} \not\xrightarrow{}}
     \end{array}
   \]
  
   \[
     \begin{array}{l@{\qquad}l}
       \inference[\texttt{varSub}]{}{v[\vec{u}] = \vec{u}(v)}
       &
         \inference[\texttt{lamSub}]{\vec{w} =
         (\oname{up}_{X,m}(\vec{u}),\var_{m+1}(\oname{new}_m)) & t[\vec{w}]         
         =t^{\prime\prime}}{(\lambda. t)[\vec{u}]=
         \lambda. t^{\prime\prime}}
     \end{array}
   \]
  
   \[
     \inference[\texttt{appSub}]{t_{1}[\vec{u}]=
       t_{1}\pr & t_{2}[\vec{u}]=t_{2}\pr}{(t_{1} \app
       t_{2})[\vec{u}]=t_{1}\pr \app t_{2}\pr}
   \]
   \caption{Law $\rho^{\cn}$ in the form of inference rules.}
   \label{fig:cn}
 \end{figure}

%

\noindent Since $\mu\Sigma\cong \Lambda$, the operational model of $\rho^\cn$ is given by a coalgebra \begin{equation}\label{eq:op-model-lambda}
  \iota^\clubsuit=\langle \gamma_{1},\gamma_{2}\rangle \c \Lambda \to
  \llangle \Lambda,\Lambda \rrangle \times (\Lambda + \Lambda^\Lambda+1).
\end{equation}
The following two propositions assert that it
matches the intended model described in
\eqref{eq:lambda-op-model}, that is, its first component correctly
exposes the substitution structure of $\lambda$-terms and its second
component yields the transition system $\to$ on $\lambda$-terms
derived from the operational semantics in \Cref{fig:lambda}.

\begin{proposition}\label{prop:gamma1}
For every $m,n\in\fset$, $t\in \Lambda(n)$ and $\vec{u}\in \Lambda(m)^n$, we have
\[ \gamma_1(t)(\vec{u}) = t[\vec{u}]. \]
\end{proposition}

\begin{proposition}\label{prop:gamma2}
For every $n\in\fset$ and $t\in \Lambda(n)$, exactly one the following statements holds:
\begin{enumerate}
\item $\gamma_2(t) \in \Lambda(n)$ and $t\to \gamma_2(t)$;
\item $\gamma_2(t)\in \Lambda^\Lambda(n)$, $t=\lambda x.t'$ for some $t'$, and $\gamma_2(t)(e)=t'[e/x]$ for every $e\in \Lambda(n)$;
\item $\gamma_2(t)=\ast$ and $t=x\app s_1\app\cdots\app s_k$ for some $x\in V(n)$ and $s_1,\ldots, s_k\in \Lambda(n)$, where $k\geq 0$.
\end{enumerate}
\end{proposition}
Let ${\sim^{\Lambda}} \seq \Lambda \times \Lambda$ be the bisimilarity
relation (i.e.\ the greatest bisimulation) on the coalgebra \eqref{eq:op-model-lambda}.
It turns out that $\sim^\Lambda$ matches the strong variant of applicative bisimilarity~\cite{Abramsky:lazylambda}:
\begin{definition}\label{def:strong-app}
  \emph{Strong applicative bisimilarity} is the greatest relation $\sim^\ap_0 \,\subseteq\, \Lambda(0) \times \Lambda(0)$ on the set of closed $\lambda$-terms such that for 
  $t_{1} \sim^\ap_0 t_{2}$ the following conditions hold:
  \begin{align*}
    t_{1} \to t_{1}\pr
    & \;\implies\; \exists 
    t_{2}\pr.~t_{2} \to t_{2}\pr \wedge t_{1}\pr \sim^\ap_0 t_{2}\pr;
    \tag*{(A1)} \\
    t_{1} = \lambda x.t_1'
    &\;\implies\; \exists t_2'.\,t_{2} = \lambda x.t_2' \wedge \forall e \in
    \Lambda(0).~t_1'[e/x]\sim^\ap_0t_2'[e/x];
    \tag*{(A2)} \\
    t_{2} \to t_{2}\pr
    &\;\implies\; \exists t_{1}\pr.~t_{1} \to t_{1}\pr \wedge t_{1}\pr \sim^\ap_0 t_{2}\pr;
    \tag*{(A3)} \\
    t_{2} = \lambda x.t_2'
    &\;\implies\; \exists t_1'.\,t_{1} = \lambda x.t_1' \wedge \forall e \in
    \Lambda(0).~t_1'[e/x]\sim^\ap_0t_2'[e/x].\tag*{(A4)}
  \end{align*}
The \emph{open extension} of strong applicative bisimilarity is the relation $\sim^{\ap}\,\seq \Lambda\times \Lambda$ whose component $\sim^{\ap}_n\,\seq \Lambda(n)\times \Lambda(n)$ for $n>0$ is given by
\[ t_1 \sim^{\ap}_n t_2 \qquad \text{iff}\qquad t_1[\vec{u}] \sim^{\ap}_0 t_2[\vec{u}]\quad \text{for every $\vec{u}\in \Lambda(0)^n$}.\]
\end{definition}

\begin{proposition}\label{prop:bisim-vs-appbisim}
Bisimilarity coincides with the open extension of strong applicative bisimilarity:
\[ {\sim^\Lambda}\; =\, {\sim^\ap}. \]
\end{proposition}
It is not difficult to verify that the present setting satisfies the conditions of our general compositionality result (\Cref{th:main}). Thus, from the latter and \Cref{prop:bisim-vs-appbisim} we get

\begin{corollary}
  \label{cor:cong}
 The open extension $\sim^\ap$ of strong applicative bisimilarity is a congruence.
\end{corollary}
While the above results are for the call-by-name
$\lambda$-calculus (\Cref{fig:lambda}), the call-by value $\lambda$-calculus
(\Cref{fig:cbv}) can be treated in an analogous manner.
\begin{figure}
  \[
    \begin{array}{l@{\qquad}l@{\qquad}l}
     \inference[\texttt{app1}]{ \quad q = \lambda x.\_}
        {\goes{(\lambda x.p) \app q}{p[q/x]}}
      & \inference[\texttt{app2}]{\goes{p}{p\pr}}{\goes{p \app q}{p\pr \app
        q}}
      & \inference[\texttt{app3}]{\goes{q}{q\pr}}{\goes{(\lambda x.p) \app q}{(\lambda x.p) \app q\pr}}
    \end{array}
  \]
  \caption{Small-step operational semantics of the call-by-value $\lambda$-calculus.}
  \label{fig:cbv}
\end{figure}%
The corresponding higher-order GSOS law differs from the one in \Cref{def:lamgsos}
only in the case of application on closed terms.
\begin{definition}[$V$-pointed higher-order GSOS law of the call-by-value $\lambda$-calculus]
  \begin{align*}
    \rho^{\cv}_{X,Y} \c \quad & \Sigma(jX \times B(jX,Y)) & \to \quad & B(jX, \Sigma^\star (jX+Y)) \\
    \rho^{\cv}_{X,Y,n}(tr) = \quad & \texttt{case}~tr~\texttt{of} \\
  & v \in V(n) & \mapsto \quad & \pi(v),* & \\
  & \mathsf{\lambda}.(t, f,\_) & \mapsto
    \quad & \llangle X, \lambda.(-)\comp \eta \comp \inr \rrangle (\rho_{1}(f)),
   (\eta \comp \inr)^{X}(\rho_{2}(f)) & \\
  & (t_{1}, g, t_{1}\pr) \app (t_{2}, h,\_) & \mapsto
    \quad & \lambda \vec{u}. ( g_{m}(\vec{u}) \app h_{m}(\vec{u}) ),t_{1}\pr \app t_{2} & \\
    & (t_{1}, g, k) \app (t_{2}, h,t_{2}\pr) & \mapsto
      \quad & \lambda \vec{u}. (g_{m}(\vec{u}) \app h_{m}(\vec{u})),t_{1} \app t_{2}\pr & \\
  & (t_{1}, g, k) \app (t_{2}, h,\_) & \mapsto
    \quad & \lambda \vec{u}. (g_{m}(\vec{u}) \app h_{m}(\vec{u})),\eta \comp \inr \comp k(t_{2}) & \\
  & (t_{1}, g, *) \app (t_{2}, h,\_) & \mapsto
    \quad & \lambda \vec{u}. (g_{m}(\vec{u}) \app h_{m}(\vec{u})),* &
  \end{align*}
  where $t\in \delta X(n)$, $f\in \delta\llangle X,Y\rrangle(n)$, $g,h\in
  \llangle X,Y\rrangle(n)$, $\vec{u} \in X(m)^{n}$ for $m \in \mathbb{N}$, $k\in
  Y^X(n)$, $t_1, t_2\in X(n)$ and $t_1',t_{2}'\in Y(n)$ (once again
  the brackets around the pairs on the right are omitted).
\end{definition}
Applying \Cref{th:main} for the call-by-value $\lambda$-calculus shows
that coalgebraic bisimilarity, as expressed in \Cref{prop:bisim}, is a
congruence. However, unlike the case for the call-by-name
$\lambda$-calculus, coalgebraic bisimilarity does not correspond to a
strong version of call-by-value applicative
bisimilarity~\cite{pitts_2011}: The former relates terms if they
exhibit the same behaviour when applied to arbitrary closed terms,
while the latter considers only application to \emph{values}. Capturing call-by-value applicative
bisimilarity in the coalgebraic framework is left as an open problem.





\section{Conclusions and Future Work}
\label{sec:concl}

We have introduced the notion of \emph{pointed higher-order GSOS law},
effectively transferring the principles behind Turi and Plotkin's
bialgebraic framework~\shortcite{DBLP:conf/lics/TuriP97} to higher-order
languages. We have demonstrated that, under mild assumptions, bisimilarity in
systems given as
pointed higher-order GSOS laws is a congruence, a result
guaranteeing the compositionality of semantics within our abstract
framework. In addition, we have implemented the $\SKI$ calculus as
well as the call-by-name and call-by-value $\lambda$-calculus as
pointed higher-order GSOS laws in suitable categories.

Currently, we have a general compositionality result for coalgebraic
bisimilarity which, in the case of the $\lambda$-calculus, amounts to
a strong variant of Abramsky's applicative bisimilarity. However,
applicative bisimilarity typically abstracts away from
$\beta$-reductions, treating them as invisible steps. This is an
example of \emph{weak} bisimilarity. We aim to extend our
compositionality result to weak bisimilarity, potentially involving
more restrictive versions of $\Pt$-pointed higher-order GSOS laws. One
possibility would be to look at comparable work on first-order
languages, for instance the \emph{cool} rule formats of
\citet{DBLP:journals/tcs/Bloom95} and
\citet{DBLP:journals/tcs/Glabbeek11} or the more abstract approach by
\citet{DBLP:conf/mfcs/0001WNDP21}, aimed at bialgebraic semantics. We
expect that abstract congruence results for higher-order weak
bisimilarity require the development of new proof techniques, such as
a categorical version of Howe's method~\cite{DBLP:conf/lics/Howe89,
  DBLP:journals/iandc/Howe96} suited for use with higher-order GSOS
laws, possibly taking inspiration from recent work on
Howe's method in the context of familial monads~\cite{DBLP:journals/lmcs/HirschowitzL22,DBLP:conf/lics/BorthelleHL20}.

\begin{sloppypar}
Another relevant direction is the extension of our framework to programming
languages with effects, e.g.\ nondeterministic, probabilistic, or stateful
versions of the $\lambda$-calculus. A powerful approach to
compositionality results for such languages is given by environmental
bisimulations~\cite{DBLP:conf/lics/SangiorgiKS07}, which we aim to investigate
from the perspective of higher-order abstract GSOS.
\end{sloppypar}

Supporting typed programming languages is a further significant step
towards a unifying, compositional framework based on abstract GSOS.
In recent work, \citet{DBLP:conf/lics/ArkorF20} build on the ideas by
\citet{DBLP:conf/lics/FiorePT99} to give an abstract, algebraic
account of simple type theory by considering presheaves over a
category of \emph{typed} cartesian contexts, as opposed to the
category~$\fset$ of untyped cartesian contexts. In future work we will
aim for a similar development as in \Cref{sec:lam} to give
higher-order GSOS laws for typed $\lambda$-calculi and other typed
languages.

\begin{sloppypar}
 Finally, another goal of interest is to extend the
notion of a \emph{morphism of distributive
  laws}~\cite{DBLP:journals/entcs/Watanabe02, DBLP:conf/calco/KlinN15}
 to $\Pt$-pointed higher-order GSOS laws, in order to model compilers
of higher-order languages that preserve semantic properties across
compilation. This idea has been previously explored for
first-order GSOS laws \cite{DBLP:conf/cmcs/0001NDP20, DBLP:conf/aplas/AbateBT21}.
\end{sloppypar}

\begin{acks}                            
  Stelios Tsampas wishes to thank Andreas Nuyts and Christian Williams for the
  numerous and fruitful discussions.
\end{acks}

%
%
\clearpage

\bibliography{../common/mainBiblio}

\clearpage
\appendix
\section{Appendix}
This appendix provides all proofs omitted in the main parts of the paper.

\subsection*{Proof of \Cref{prop:yon1}}
\begin{enumerate}
\item Removing the syntactic sugar from \Cref{def:hoformat}, we see
  that $\HO$ specifications for a signature~$\Sigma$ correspond
  bijectively to elements of the set
\begin{equation}
  \label{eq:hoformatns}
  \prod_{\f\in \Sigma}\prod_{W\seq \ar(\f)}\Bigl(\Sigma^\star \bigl(\ar(\f)+W+\ar(\f)\times \ol{W}\bigr) + \Sigma^\star \bigl(\ar(\f)+1+W+(\ar(\f)+1)\times
    \ol{W}\bigr)\Bigr).
\end{equation}
Here, we identify the natural number $\ar(\f)$ with the set $\{1,\ldots, \ar(\f)\}$, and we let $\ol{W}=\ar(\f)\smin W$ denote the complement. The idea is that the
summands under $\Sigmas$ spell out which variables may be used in the conclusion of the
respective rule. For instance, the rule \texttt{app1-b} of \Cref{ex:ski-to-ho} corresponds to the element
\[ (2,1)\in \ar(f)\times \ol{W}\seq \Sigma^\star \bigl(\ar(\f)+W+\ar(\f)\times \ol{W}\bigr) \]
where the variable $x_1^{x_2}$ is identified with $(2,1)$, $\f=\circ$ is the application operator and $W=\{2\}$.

 We are thus left to prove that elements of the set~\eqref{eq:hoformatns} are in a bijective correspondence with families maps of type~\eqref{eq:simpleho}.
\item 
  Given functors $F,G\colon \Set^{\opp}\times \Set\times \Set\to \Set$ we
  write $\DiNat_{X,Y}(F(X,X,Y),G(X,X,Y))$ for the collection of all families of
  maps
  $\rho_{X,Y}\colon F(X,X,Y)\to G(X,X,Y)$
dinatural in $X\in \Set$ and natural in $Y\in \Set$. Then we have the following
chain of bijections:
\begin{align*}
& \DiNat_{X,Y}(\Sigma(X\times B_u(X,Y)),B_u(X,\Sigmas(X+Y))) \\
=\; & \DiNat_{X,Y}\bigl(\coprod_{\f\in\Sigma} (X\times B_u(X,Y))^{\ar(\f)},B_u(X,\Sigmas(X+Y))\bigr) \\
\cong\; & \prod_{\f\in \Sigma}\DiNat_{X,Y}\bigl((X\times B_u(X,Y))^{\ar(\f)},B_u(X,\Sigmas(X+Y))\bigr) \\
=\; & \prod_{\f\in \Sigma}\DiNat_{X,Y}\bigl((X\times (Y + Y^{X}))^{\ar(\f)}, B_u(X,\Sigmas(X+Y))\bigr) \\
\cong\; & \prod_{\f\in \Sigma}\DiNat_{X,Y}\bigl(X^{\ar(\f)}\times (Y + Y^{X})^{\ar(\f)}, B_u(X,\Sigmas(X+Y))\bigr) \\
\cong\; & \prod_{\f\in \Sigma}\DiNat_{X,Y}\bigl(\coprod_{W\seq \ar(\f)} X^{\ar(\f)}\times Y^W \times Y^{X\times \ol{W}} , B_u(X,\Sigmas(X+Y))\bigr) \\
\cong\; & \prod_{\f\in \Sigma} \prod_{W\seq \ar(\f)}\DiNat_{X,Y}\bigl(X^{\ar(\f)}\times Y^{W + X\times \ol{W}} , B_u(X,\Sigmas(X+Y))\bigr). 
\end{align*}
For the penultimate step we use that products distribute over coproducts in $\Set$.
We claim that 
\begin{equation}\label{eq:dinat-sum}
 \begin{gathered}
    \DiNat_{X,Y}\bigl( X^{\ar(\f)}\times Y^{W + X\times \ol{W}} , B_u(X,\Sigmas(X+Y))\bigr)      \\
    \cong \\
    \DiNat_{X,Y}\bigl( X^{\ar(\f)}\times Y^{W + X\times \ol{W}} , \Sigmas(X+Y)\bigr) \;+\; \DiNat_{X,Y}\bigl( X^{\ar(\f)}\times Y^{W + X\times \ol{W}} ,(\Sigmas(X+Y))^{X}\bigr).
  \end{gathered}
\end{equation}
To see this, let $\rho$ be a family of maps in $\DiNat_{X,Y}\bigl(X^{\ar(\f)}\times Y^{W + X\times \ol{W}} , B_u(X,\Sigmas(X+Y))\bigr)$. Consider the diagram below, where $!_X\c X\to 1$ and $!_Y\c Y\to 1$ are the unique maps. The upper part commutes by naturality in $Y$ and the lower part by dinaturality in $X$.
\[
\begin{tikzcd}
\id \times Y^{W\times X\times\ol{W}} \ar{r}{\rho_{X,Y}} \ar{d}[swap]{X^{\mathsf{ar}(\f)}\times (!_Y)^{\id+X\times\id}} & \Sigmas(X+Y)+(\Sigmas(X+Y))^X \ar{d}{\Sigmas(\id+!_Y)+ (\Sigmas(\id+!_Y))^X}  \\
 X^{\mathsf{ar}(\f)}\times 1^{W+X\times \ol{W}}  \ar{r}{\rho_{X,1}} & \Sigmas(X+1)+(\Sigmas(X+1))^X  \ar{d}{ \Sigmas(!_X+\id) + (\Sigmas(!_X+\id))^X }  \\
X^{\mathsf{ar}(\f)}\times 1^{W+1\times\ol{W}} \ar{u}{\id\times 1^{\id+!_X\times\id}  }[swap]{\cong} \ar{d}[swap]{ (!_X)^{\mathsf{ar}(\f)}\times \id} & \Sigmas(1+1)+(\Sigmas(1+1))^X \\
 1^{\mathsf{ar}(\f)}\times 1^{W+1\times\ol{W}} \ar{r}{\rho_{1,1}} & \Sigmas(1+1)+(\Sigmas(1+1))^1 \ar{u}[swap]{\id+(\Sigmas(1+1))^{!_X}  }    
\end{tikzcd}
\] 
Since the map $\rho_{1,1}$ has domain $1^{\mathsf{ar}(\f)}\times 1^{W+1\times\ol{W}}\cong 1$, it factorizes through one of the summands of its codomain. The commutativity of the above diagram then implies that also $\rho_{X,Y}$ factorizes through the corresponding summand of its codomain, which proves \eqref{eq:dinat-sum}.
\item It remains to show that the two summands in \eqref{eq:dinat-sum} are isomorphic to the
corresponding summands in~\eqref{eq:hoformatns}. For the first one, letting $\NT_X(F(X),G(X))$ denote the collection of natural transformations between functors $F,G\c \Set\to \Set$, we compute 
\begin{align*}
  & \DiNat_{X,Y}\bigl(X^{\ar(\f)}\times Y^{W + X\times \ol{W}} , \Sigmas(X+Y)\bigr)\\
  \cong\;& \NT_X\bigl(X^{\ar(\f)}, \NT_Y(Y^{W + X\times \ol{W}}, \Sigmas(X+Y))\bigr) \\
  \cong\; &  \NT_Y\bigl(Y^{W + \ar(\f)\times \ol{W}}, \Sigmas(\ar(\f)+Y)\bigr) \\
  \cong\; &  \Sigmas(\ar(\f)+W+\ar(\f)\times \ol{W}).
\end{align*}
The last two isomorphisms use the Yoneda lemma, and the first one is given by currying:
\[ \rho \quad\mapsto \quad (\,\lambda x\in X^{\ar(\f)}. \,(\rho_{X,Y}(x,-))_Y\,)_X.\]
Note that for every $x\in X^{\ar(\f)}$ the family $(\rho_{X,Y}(x,-)\c Y^{W+X\times \ol{W}}\to \Sigmas(X+Y))_Y$ is natural in $Y$; the naturality squares are equivalent to the ones witnessing naturality of $\rho_{X,Y}$ in $Y$. Similarly, the family $(\,\lambda x\in X^{\ar(\f)}. \,(\rho_{X,Y}(x,-))_Y\,)_X$ is natural in $X$; the naturality squares are equivalent to the commutative hexagons witnessing dinaturality of $\rho_{X,Y}$ in $X$. 

Much analogously, we have
\begin{align*}
  & \DiNat_{X,Y}\bigl(X^{\ar(\f)}\times Y^{W + X\times \ol{W}} , (\Sigmas(X+Y))^X\bigr) \\
  \cong\;& \NT_X\bigl(X^{\ar(\f)}\times X, \NT_Y(Y^{W + X\times \ol{W}}, \Sigmas(X+Y))\bigr) \\
  \cong\; &  \NT_Y\bigl(Y^{W + (\ar(\f)+1)\times \ol{W}}, \Sigmas(\ar(\f)+1+Y)\bigr) \\
  \cong\; &  \Sigmas(\ar(\f)+1+W+(\ar(\f)+1)\times \ol{W}).
\end{align*}
This concludes the proof.
\end{enumerate}

\subsection*{Proof of \Cref{lem:clubs}}
Existence of $a^\clubsuit$ follows from the fact that the initial algebra morphism
$\brks{w,\,a^\clubsuit}$, determined by the diagram
\begin{equation}\label{diag:w-a-clubs}
\begin{tikzcd}[column sep=7ex, row sep=normal]
\Sigma(\mS) 
  \dar["\Sigma\brks{w,\,a^\clubsuit}"']
  \ar[rrr,"\iota"] 
  &[1.5ex] &[2ex] &
\mS 
  \dar["\brks{w,\,a^\clubsuit}"]
  \\
\Sigma (A\times {B(A,A)})
  \rar["\brks{a\cdot\Sigma\fst,\,\rho_{A,A}}"] 
&
A\times B(A,\Sigma^\star(A+A))
  \rar["{\id\times B(\id,\Sigmas\nabla)}"] &
A\times B(A,\Sigma^\star A)
  \rar["{\id\times B(\id, \hat a)}"] &
A\times B(A,A)
\end{tikzcd}
\end{equation}
indeed yields a suitable $a^\clubsuit$, for which the diagram of interest commutes; note that $w\c \mS\to (A,a)$ is a $\Sigma$-algebra morphism by the above diagram, hence $w=\iter a$. To show 
uniqueness, suppose that~ $a^\clubsuit$ with the requisite property exists. Then Diagram~\eqref{diag:w-a-clubs} 
commutes with $w=\iter a$, so uniqueness of $a^\clubsuit$ is entailed by 
the uniqueness of $\brks{w,a^\clubsuit}$.

\subsection*{Details for \Cref{rem:gsos-to-ho-gsos}}
We show that $\rho$ as defined in \Cref{rem:gsos-to-ho-gsos} is a (0-pointed) higher-order GSOS law of $\Sigma$ over $B$, where $B(X,Y)=FY$. Naturality of $\rho_{X,-}$ is shown by the following commutative diagram for $f\c Y\to Y$'. The left part obviously commutes, and the right part commutes by naturality of $\lambda$.
\[
\begin{tikzcd}[column sep=.2em, row sep=3em]
\Sigma(X\times B(X,Y)) = \Sigma(X\times FY) \ar{d}[swap]{\Sigma(\id\times Ff)} \ar{r}[yshift=0.5em]{\Sigma(\inl\times F\inr)} & \Sigma((X+Y)\times F(X+Y)) \ar{d}{\Sigma((\id+f)\times F(\id+f))} \ar{r}[yshift=.5em]{\lambda_{X+Y}} &  F\Sigmas(X+Y)=B(X,\Sigmas(X+Y)) \ar{d}{F\Sigmas(\id+f)} \\
\Sigma(X\times B(X,Y')) = \Sigma(X\times FY') \ar{r}[yshift=0.5em]{\Sigma(\inl\times F\inr)} & \Sigma((X+Y')\times F(X+Y')) \ar{r}[yshift=.5em]{\lambda_{X+Y'}} &  F\Sigmas(X+Y')=B(X,\Sigmas(X+Y'))
\end{tikzcd}
\]
Since $B(X,Y)=FY$ does not depend on its contravariant component, dinaturality
 of $\rho_{-,Y}$ is equivalent to naturality and is shown by the commutative diagram below for  $g\c X\to X'$:
\[
\begin{tikzcd}[column sep=.2em, row sep=3em]
\Sigma(X\times B(X,Y)) = \Sigma(X\times FY) \ar{d}[swap]{\Sigma(g\times \id)} \ar{r}[yshift=0.5em]{\Sigma(\inl\times F\inr)} & \Sigma((X+Y)\times F(X+Y)) \ar{d}{\Sigma((g+\id)\times F(g+\id))} \ar{r}[yshift=.5em]{\lambda_{X+Y}} &  F\Sigmas(X+Y)=B(X,\Sigmas(X+Y)) \ar{d}{F\Sigmas(g+\id)} \\
\Sigma(X'\times B(X',Y)) = \Sigma(X'\times FY) \ar{r}[yshift=0.5em]{\Sigma(\inl\times F\inr)} & \Sigma((X'+Y)\times F(X'+Y)) \ar{r}[yshift=.5em]{\lambda_{X'+Y}} &  F\Sigmas(X'+Y)=B(X',\Sigmas(X'+Y))
\end{tikzcd}
\]
We now show that the laws $\lambda$ and $\rho$ are semantically equivalent, that is, their  operational models
\[ \gamma\c \mS\to F(\mS)\qquad\text{and}\qquad \iota^\clubsuit\c \mS\to F(\mS)=B(\mS,\mS)  \]
coincide. By definition (see \Cref{sec:abstract-gsos}), the coalgebra structure $\gamma$ is uniquely determined by the following commutative diagram:
\begin{equation*}
\begin{tikzcd}[column sep=8ex, row sep=normal]
\Sigma(\mS) 
  \dar["\Sigma \brks{\id,\,\gamma}"']
  \ar[rrr,"\iota"] 
  &[-2ex] &[1ex] &[2ex]
\mS 
  \dar["\gamma"]
  \\
\Sigma (\mS\times F(\mS))
  \ar{rr}{\lambda_{\mS}} 
&
&
F\Sigmas(\mS) \ar{r}{F\hat\ini}
&
F(\mS)
\end{tikzcd}
\end{equation*} 
Thus, we only need to show that $\iota^\clubsuit$ is such a $\gamma$, which follows from the commutative diagram below. The upper cell commutes by definition of $\iota^\clubsuit$, the cell involving $\lambda_{\mS}$ commutes by naturality of $\lambda$, and the remaining cells commute either trivially or by definition.
\tikzcdset{scale cd/.style={every label/.append style={scale=#1},
    cells={nodes={scale=#1}}}}
\[ 
\begin{tikzcd}[scale cd=.8, column sep=1em]
\Sigma(\mS) \ar{rrrrr}{\iota} \ar{ddd}[swap]{\Sigma\langle \id,\,\iota^\clubsuit\rangle} \ar{dr}{\langle \iter \iota,\, \iota^\clubsuit\rangle} & & & & & \mS \ar{ddd}{\iota^\clubsuit} \ar{dl}[swap]{\iota^\clubsuit} \\
& \Sigma(\mS\times B(\mS,\mS)) \ar{d}{\Sigma(\inl\times F\inr)} \ar{r}{\rho_{\mS,\mS}} \ar[equals, bend right=2em]{ddl} & B(\mS,\Sigmas(\mS+\mS)) \ar{r}[yshift=.5em]{B(\id,\Sigmas\nabla)} \ar{dd}{F\Sigmas\nabla} & B(\mS,\Sigmas(\mS)) \ar{r}[yshift=.5em]{B(\id,\hat\ini)} & B(\mS,\mS) \ar[equals]{ddr} & \\
& \Sigma((\mS+\mS)\times F(\mS+\mS)) \ar{ur}[swap]{\lambda_{\mS+\mS}} \ar{dl}{\Sigma(\nabla\times F\nabla)} & & & & \\
\Sigma(\mS\times F(\mS)) \ar{rr}{\lambda_{\mS}} & & F\Sigmas(\mS) \ar{rrr}{F\hat\ini} & & & F(\mS) 
\end{tikzcd}
\]

\subsection*{Proof of \Cref{prop:rho-bialg-cat}}
It is clear that $\id_A$ is a $\rho$-bialgebra morphism from $(A,a,c)$ to $(A,a,c)$. It remains to show that the composite $h\comp g$ of $\rho$-bialgebra morphisms $g\c A\to A'$ and $h\c A'\to A''$ is again a $\rho$-bialgebra morphism; this follows from commutation of the diagram below.
\begin{equation*}
\begin{tikzcd}[column sep=6ex, row sep=normal, baseline = (B.base)]
A
  \ar[rrr,"c"]
  \dar["g"'] & &[4ex] &[4ex]
B(A,A)
  \dar["{B(\id,g)}"] \\
A'
  \ar[rr,"c'"]
  \dar["h"'] & &
B(A',A')
 \ar[r, "{B(g,\id)}"]
 \dar["{B(\id,h)}"'] &
B(A,A')
  \dar["{B(\id,h)}"] \\
A''
  \ar[r,"c''"] 
& 
B(A'',A'')
  \ar[r,"{B(h,\id)}"] 
&
B(A',A'')
  \ar[r,"{B(g,\id)}"] 
&|[alias=B]|
B(A,A'')
\end{tikzcd}
\end{equation*}

\subsection*{Proof of \Cref{prop:initial-bialgebra}}
It follows directly by  
definition of~$\iota^\clubsuit$ according to \autoref{lem:clubs}, and by observing that $\iter \ini = \id$,
that $(\mS,\iota,\iota^\clubsuit)$ is a $\rho$-bialgebra. To prove initiality, suppose that $(A,a,c)$ is a $\rho$-bialgebra. We show that $\iter a\c\mS\to A$
is the unique $\rho$-bialgebra morphism from $(\mS,\iota,\iota^\clubsuit)$ to $(A,a,c)$. To show that $\iter a$ is a $\rho$-bialgebra morphism, we need to verify that the following diagram commutes:
\begin{equation*}
\begin{tikzcd}[column sep=10.2ex, row sep=normal]
\mS
  \ar[rr,"\iota^\clubsuit"]
  \dar["\iter a"']
  \ar[dr,"a^\clubsuit"] & &
B(\mS,\mS)
  \dar["{B(\id,\iter a)}"] 
\\
A\ar[r,"c"] 
& 
B(A,A)
  \rar["{B(\iter a,\id)}"]
& 
B(\mS,A)
\end{tikzcd}
\end{equation*}
The quadrangular cell commutes by~\autoref{lem:law_comm}, and we are left to
show that $c\comp (\iter a) = a^\clubsuit$.
This follows from the fact that $c\comp (\iter a)$ satisfies the
characteristic property of $a^\clubsuit$. Indeed, the diagram
\begin{equation*}
\begin{tikzcd}[column sep=5ex, row sep=normal]
\Sigma\mS
  \dar["\Sigma(\iter a)"']
  \ar[rrr,"\iota"] 
  \ar[shiftarr = {xshift=-45}, swap]{dd}{\Sigma\brks{\iter a,\,c\comp\iter a}}
& &[2ex] &[1ex]
\mS
  \dar["\iter a"]
  \ar[shiftarr = {xshift=28}]{dd}{c\comp \iter a}
\\
\Sigma A
  \dar["\Sigma\brks{\id,\,c}"']
  \ar[rrr,"a"] 
& & &
A
  \dar["c"]
\\
\Sigma(A\times {B(A,A)})
  \rar["{\rho_{A,A}}"] 
&
B(A,\Sigma^\star A)
  \rar["{B(\id,\Sigmas\nabla)}"] 
&
B(A,\Sigma^\star (A+A))
  \rar["{B(\id,\hat{a})}"] 
&
B(A,A)
\end{tikzcd}
\end{equation*}
commutes: the top cell commutes by definition of $\iter a$, and the bottom one 
commutes by the assumption that $(A,a,c)$ is a $\rho$-bialgebra.

Uniqueness of the bialgebra morphism $\iter a\c\mS\to A$ is by initiality of $\mS$ as a $\Sigma$-algebra, since every bialgebra morphism is, by definition, in particular a $\Sigma$-algebra morphism.

\subsection*{Proof of \Cref{prop:sim-bialg}}
Note that the outside of the diagram
\begin{equation*}
\begin{tikzcd}[column sep=5ex, row sep=normal]
\Sigma\mS
  \dar["\Sigma(\iter\iotaq)"']
  \ar[rr,"\iota"] 
  \ar[shiftarr = {xshift=-60}, swap]{dd}{\Sigma\brks{\iter\iotaq,\,\iotaq^\clubsuit}}
& &[3ex] 
\mS
  \dar["\iter\iotaq"]
  \ar[shiftarr = {xshift=40}]{dd}{\iotaq^\clubsuit}
\\
\Sigma\mSq
  \dar["\Sigma\brks{\id,\,\varsigma}"']
  \ar[rr,"\iotaq"] 
& & 
\mSq
  \dar["\varsigma"]
\\
\Sigma(\mSq\times {B(\mSq,\mSq)})
  \rar["{\rho_{\mSq,\mSq}}"]
&
B(\mSq,\Sigmas(\mSq+\mSq))
  \rar["{B(\id,\widehat\iotaq\comp\Sigmas\nabla)}"] 
&
B(\mSq,\mSq)
\end{tikzcd}
\end{equation*}
commutes by definition of $\iotaq^\clubsuit$. The side cells commute by~\autoref{lem:mSq-coalg},
and the top middle cell commutes by definition of $\iter\iotaq$. Note that $\Sigma(\iter\iotaq)$ is a coequalizer,
since $\iter\iotaq$ is a reflexive coequalizer and $\Sigma$ preserves it. Hence $\Sigma(\iter\iotaq)$ is 
epic, and therefore the bottom middle cell commutes, which is the $\rho$-bialgebra
law in question.

\subsection*{Proof of \Cref{lem:final-coalgebra}}
  The functor $B(X,-)$ preserves limits of $\omega^\opp$-chains: the
  right adjoints $\llangle X,- \rrangle$ and $(-)^X$ preserve all limits, and
  limits of $\omega^\opp$-chains commute with products and coproducts
  in $\Set$ and thus in $\vcat$ (using that limits and colimits in
  presheaf categories are formed pointwise). Therefore, dually
  to the classic result by \citet{adamek74}, the functor $B(X,-)$ has a final coalgebra
  computed as the limit of the final $\omega^\opp$-chain
  \[
    1\leftarrow B(X,1) \leftarrow B(X,B(X,1)) \leftarrow B(X,B(X,B(X,1)))  \leftarrow \cdots
  \]

\subsection*{Proof of \Cref{prop:bisim}}
For the $\Longrightarrow$ direction, suppose that $R$ is a bisimulation, i.e.\ there exists a coalgebra structure $\langle r_1,r_2\rangle$ on $R$ making the diagram below commute, where $p_1,p_2$ are the projections: 
\[
\begin{tikzcd}[column sep=45, row sep=3em]
W \ar{d}[swap]{\langle c_1,c_2\rangle} & R \ar{l}[swap]{p_1} \ar{r}{p_2} \ar{d}{\langle r_1,r_2\rangle} & Z \ar{d}{\langle d_1,d_2\rangle}  \\
\llangle X,W\rrangle\times (W+W^X+1) &  \ar{l}[swap,yshift=.5em]{\llangle X,p_1\rrangle\times (p_1+p_1^X+1) } \llangle X,R\rrangle\times(R+R^X+1) \ar{r}[yshift=.5em]{\llangle X,p_2\rrangle\times (p_2+p_2^X+1) }  & \llangle X,Z\rrangle \times (Z+Z^X+1)
\end{tikzcd}
\]
Then (1) holds because $R$ is a sub-presheaf of $W\times Z$, and (2), (3) and (5) are immediate from the above diagram. Concerning (4), let  $w \mathbin{R(n)} z$ and suppose that $c_2(w)=:f\in W^X(n)$. Then the above diagram implies that $r_2(w,z)=:h\in R^X(n)$ and $d_2(z)=:g\in Z^X(n)$. Moreover, for all $e\in X(n)$, 
\[ f(e) = \ev(f,e) = p_1(\ev(h,e)) \mathbin{R(n)} p_2(\ev(h,e)) = \ev(g,e) = g(e) \]
where the second and the penultimate equality follow via naturality of $\ev$.

For the $\Longleftarrow$ direction, suppose that $R(n)\seq W(n)\times Z(n)$, $n\in\fset$, is a family of relations satisfying (1)--(8) for all $w\mathbin{R(n)} z$. Condition (1) asserts that $R$ is a sub-presheaf of $W\times Z$; thus it remains to define a coalgebra structure $\langle r_1,r_2\rangle$ on $R$ making the diagram above commute. It suffices to define the components 
\[\langle r_{1,n},r_{2,n}\rangle\c  R(n)\to \llangle X,R\rrangle(n)\times (R(n)+ R^X(n)+1)\] 
and prove that the diagram commutes pointwise at every $n\in \fset$; the naturality of $r_1,r_2$ then follows since the two lower horizontal maps in the diagram are jointly monomorphic.

We define 
\[r_{1,n}\c R(n)\to \llangle X,R\rrangle(n)=\NT(X^n,R)\qquad\text{by}\qquad
 r_{1,n}(w,z) = \langle c_{1,n}(w), d_{1,n}(z)\rangle.\]
Condition (2) shows that this map is well-typed and that it makes the first component of the diagram commute. 

 To define $r_{2,n}$, using extensivity of the presheaf topos $\vcat$ we express $R$ as a coproduct  $R=R_0+R_1+R_2$ of the sub-presheaves given by
\begin{align*}
R_0(n) & = \{\, (w,z)\in R(n) : c_2(w)\in W(n),\, d_2(z)\in Z(n) \,\}, \\
R_1(n) & = \{\, (w,z)\in R(n) : c_2(w)\in W^X(n),\, d_2(z)\in Z^X(n) \,\}, \\
R_2(n) & = \{\, (w,z)\in R(n) : c_2(w)=\ast,\, d_2(z)=\ast \,\}.
\end{align*}
Thus, it suffices to define $r_{2,n}\c R_0(n)+R_1(n)+R_2(n)\to R(n)+R^X(n)+1$ separately for each summand of its domain. Given $(w,z)\in R_0(n)$, we put
\[ r_{2,n}(w,z) = (c_2(w),d_2(z))\in R(n). \]
By condition (3), this is well-typed and makes the second component of the diagram (with domain $R$ restricted to $R_0$) commute. Similarly, for $(w,z)\in R_2(n)$ we put 
\[ r_{2,n}(w,z)=\ast; \]
the second component of the diagram (with domain $R$ restricted to $R_2$) then commutes by condition (5). Finally, for $(w,z)\in R_1(n)$ we put
\[ r_{2,n}(w,z) = \curry\, h (w,z) \in R^X(n)\]
where $h\colon R_1\times X\to R$ is the natural transformation whose component at $m\in \fset$ is given by
\[ h_m((w',z'),e) = (c_2(w')(e), d_2(z')(e)) \in R(m).  \] 
Condition (4) asserts that $h_m$ is well-typed and that the second component of the diagram (with domain $R$ restricted to $R_1$) commutes.

\begin{rem}\label{rem:op-model-lambda}
The operational model of the higher-order GSOS law $\rho^\cn$ is the coalgebra 
\[
  \iota^\clubsuit=\langle \gamma_{1},\gamma_{2}\rangle \c \Lambda \to
  \llangle \Lambda,\Lambda \rrangle \times (\Lambda + \Lambda^\Lambda+1)
\] 
uniquely determined the following commutative diagram, see \Cref{lem:clubs}:
\[
\begin{tikzcd}[column sep=63]
V+\delta\Lambda+\Lambda^2 \ar{ddd}{ \iota=[\var,\lambda.(-),\circ]} \ar{r}[yshift=.5em]{V+\delta\langle \id,\gamma_1,\gamma_2\rangle + \langle \id,\gamma_1,\gamma_2\rangle^2} & V+\delta(\Lambda\times\llangle \Lambda,\Lambda\rrangle \times(\Lambda+\Lambda^\Lambda+1)) + ( \Lambda\times\llangle\Lambda,\Lambda\rrangle \times (\Lambda+\Lambda^\Lambda+1))^2  \ar{d}{\rho_{\Lambda,\Lambda}^{\cn}} \\
&  \llangle \Lambda,\Sigmas(\Lambda+\Lambda)\rrangle \times (\Sigmas(\Lambda+\Lambda) + (\Sigmas(\Lambda+\Lambda))^\Lambda +1)  \ar{d}{\llangle \id,\Sigmas\nabla\rrangle \times (\Sigmas\nabla + (\Sigmas\nabla)^\Lambda +\id) }  \\
&  \llangle \Lambda,\Sigmas\Lambda\rrangle \times (\Sigmas\Lambda + (\Sigmas\Lambda)^\Lambda +1) \ar{d}{ \llangle \id,\hat\iota\rrangle \times (\hat\iota + {\hat\iota}^\Lambda +\id) } \\
\Lambda \ar{r}{ \langle\gamma_1,\gamma_2\rangle} &  \llangle \Lambda,\Lambda\rrangle \times (\Lambda+ \Lambda^\Lambda +1) 
 \end{tikzcd}
\]
\end{rem}

\subsection*{Proof of \Cref{prop:gamma1}}
We proceed by induction on the structure of $t$.
  \begin{itemize}
  \item For $t=v \in V(n)$,  
\[\gamma_1(v)(\vec{u})=\pi(v)(\vec{u}) = u_v = v[\vec{u}];\]
the first equality follows from the definition of $\gamma_1$ (see \Cref{rem:op-model-lambda}), the second one from the definition of $\pi$, and the third one from the definition of substitution.
  \item For $t = \lambda x.t\pr$ (where $x=n$ and $t'\in \Lambda(n+1)$),  
  \[\gamma_{1}(t)(\vec{u}) =
    \lambda m.(\gamma_{1}(t\pr)(\oname{up}_{\Lambda,m}(\vec{u}),m)) =  \lambda m.(t\pr[\oname{up}_{\Lambda,m}(\vec{u}),m]) = t[\vec{u}]; \]
the first equality uses the definition of $\gamma_1$, the second one follows by induction and the third one by the definition of substitution.
  \item For $t = t_{1} \app t_{2}$,
\[ \gamma_{1}(t)(\vec{u}) = \gamma_{1}(t_{1})(\vec{u}) \app
    \gamma_{1}(t_{2})(\vec{u}) = t_{1}[\vec{u}] \app t_{2}[\vec{u}] = t[\vec{u}]; \]
 the first equality follows from the definition of $\gamma_1$, the second one by induction, and the third one by the definition of substitution.
  \end{itemize}

\subsection*{Proof of \Cref{prop:gamma2}}
Since the three statements are mutually exclusive, it suffices to show that at least one of them holds. We proceed by induction on the structure of $t$:
   \begin{itemize}
  \item For $t=v \in V(n)$, 
we have $\gamma_2(t)=\ast$ by definition of $\gamma_2$ (see \Cref{rem:op-model-lambda}), so (3) holds (put $x=v$ and $k=0$).
  \item For $t = \lambda x.t\pr$, we have $\gamma_2(t)\in \Lambda^\Lambda(n)$ by definition of $\rho^{\cn}$. Moreover, for every $e\in \Lambda(n)$,
\[ \gamma_2(t)(e) = \rho_2(\gamma_1(t'))(e) = \gamma_1(t')(0,\ldots,n-1,e) = t'[e/x] \]
where the first two equalities use the definition of $\gamma_2$ and $\rho_2$, respectively, and the third one follows from \Cref{prop:gamma1}. Thus (2) holds.
  \item For $t = t_{1} \app t_{2}$, we distinguish three cases:
\begin{itemize}
\item If $t_1=x\in V(n)$, then $\gamma_2(t)=\ast$ and $t=x\app t_2$, so (3) holds.
\item If $t_1=\lambda x.t_1'$, we have $t\to t':= t_1'[t_2/x]$. Then  
\[ \gamma_2(t) = \gamma_2(t_1)(t_2) = t_1'[t_2/x] = t' \]
where the first equality uses the definition of $\gamma_2$ and the second one follows by induction. Thus (1) holds.
\item If $t_1=t_{1,1}\app t_{1,2}$, then by induction either (1) or (3) holds for $t_1$. If (1) holds for $t_1$, we have $\gamma_2(t_1)\in \Lambda(n)$ and $t_1\to \gamma_2(t_1)$. By definition of $\to$ and $\gamma_2$, this implies $\gamma_2(t)\in \Lambda(n)$ and
\[ t \to \gamma_2(t_1)\app t_2 = \gamma_2(t), \]
so (1) holds for $t$. If (3) holds for $t_1$, we have $\gamma_2(t_1)=\ast$ and $t_1=x\app s_1\app\cdots\app s_k$ for some $x\in V(n)$ and $s_1,\ldots, s_k\in \Lambda(n)$. Then also $\gamma_2(t)=\ast$ and $t=x\app s_1\app\cdots\app s_k \app t_2$, proving (3) for $t$.  
\end{itemize}
  \end{itemize}

\subsection*{Proof of \Cref{prop:bisim-vs-appbisim}}
By \Cref{prop:bisim,prop:gamma1,prop:gamma2}, bisimilarity is the greatest relation $\sim^\Lambda\,\seq \Lambda\times\Lambda$ such that for every $n\in \fset$ and $t_1 \sim^\Lambda_n t_2$ the following conditions hold:
\begin{enumerate}
\item $t_1[r(0),\ldots,r(n-1)] \sim^\Lambda_m t_2[r(0),\ldots,r(n-1)]$ for all $r \c n \to m$;
\item $t_1[\vec{u}] \sim^\Lambda_m t_2[\vec{u}]$ for all $m\in \fset$ and $\vec{u}\in \Lambda(m)^n$;
\item $t_1\to t_1' \implies \exists t_2'.~ t_2\to t_2' \wedge t_1' \sim^\Lambda_n t_2'$;
\item $t_1=\lambda x.t_1' \implies \exists t_2'.t_2=\lambda x.t_2' \wedge \forall e\in \Lambda(n).  t_1'[e/x] \sim^\Lambda_n t_2'[e/x]$;
\item  $\exists x\in V(n), s_1,\ldots,s_k\in \Lambda(n).\, t_1=x\app s_1\app\cdots\app s_k \Longrightarrow \exists y\in V(n), s_1',\ldots,s_m'\in \Lambda(n).\, t_2=y\app s_1'\app\cdots\app s_m'$;
\item $t_2\to t_2' \implies \exists t_1'.~ t_1\to t_1' \wedge t_1'\sim^\Lambda_n t_2'$;
\item $t_2=\lambda x.t_2' \implies \exists t_1'.t_1=\lambda x.t_1' \wedge \forall e\in \Lambda(n).  t_1'[e/x] \sim^\Lambda_n t_2'[e/x]$;
\item  $\exists y\in V(n), s_1',\ldots,s_m'\in \Lambda(n).\, t_2=y\app s_1'\app\cdots\app s_m' \Longrightarrow  \exists x\in V(n), s_1,\ldots,s_k\in \Lambda(n).\, t_1=x\app s_1\app\cdots\app s_k$. 
\end{enumerate}
Note that condition (1) is redundant, as it follows from (2) by putting $\vec{u}=\var_m\comp r$.

\medskip
\emph{Proof of $\sim^\Lambda \,\seq\, \sim^\ap$.}  Note first that $\sim^\Lambda_0\,\seq\, \Lambda(0)\times \Lambda(0)$ is a strong applicative bisimulation: the above conditions (3), (4), (6), (7) for $n=0$ correspond precisely to (A1)--(A4) with $\sim^\ap_0$ replaced by $\sim^\Lambda_0$. It follows that $\sim^\Lambda_0\,\seq\, \sim^\ap_0$ because $\sim^\ap_0$ is the greatest strong applicative bisimulation. Moreover, for $n>0$ and $t_1\sim^\Lambda_n t_2$, we have
\[ t_1[\vec{u}] \sim^\Lambda_0 t_2[\vec{u}]\quad \text{for every $\vec{u}\in \Lambda(0)^n$}\]
by condition (2), whence 
\[ t_1[\vec{u}] \sim^\ap_0 t_2[\vec{u}]\quad \text{for every $\vec{u}\in \Lambda(0)^n$}\]
because $\sim^\Lambda_0\,\seq\, \sim^\ap_0$, and so $t_1\sim^\ap_n t_2$. This proves $\sim^\Lambda_n \,\seq\, \sim^\ap_n$ for $n>0$ and thus $\sim^\Lambda\,\seq\, \sim^\ap$ overall.

\medskip\noindent
\emph{Proof of $\sim^\ap \,\seq\, \sim^\Lambda$.} Since $\sim^\Lambda$ is the greatest bisimulation, it suffices to show that $\sim^\ap$ is a bisimulation. Thus suppose that $n\in \fset$ and $t_1\sim^\ap_n t_2$; we need to verify the above conditions (2)--(8) with $\sim^\Lambda_n$ replaced by $\sim^\ap_n$. Let us first consider the case $n=0$:
\begin{enumerate}\addtocounter{enumi}{1}
\item Since $t_1$ and $t_2$ are closed terms, this condition simply states that $t_1\sim^\ap_m t_2$ for every $m>0$. This holds by definition of $\sim^\ap_m$ because $t_1[\vec{u}]=t_1 \sim^\ap_0 t_2=t_2[\vec{u}]$ for every $\vec{u}\in \Lambda(0)^m$. 
\item holds by (A1).
\item holds by (A2). 
\item holds vacuously because $t_1$ is a closed term. 
\item holds by (A3).
\item holds by (A4). 
\item holds vacuously because $t_2$ is a closed term. 
\end{enumerate} 
Now suppose that $n>0$: 
\begin{enumerate}
\item[(2)] Let $\vec{u}=(u_0,\ldots,u_{n-1})\in \Lambda(m)^n$. If $m=0$ we have $t_1[\vec{u}]\sim^\ap_0 t_2[\vec{u}]$ by definition of $\sim^\ap_n$. If $m>0$ and   $\vec{v}\in \Lambda(0)^m$ we have
\[ t_1[\vec{u}][\vec{v}] = t_1[u_0[\vec{v}],\ldots, u_{n-1}[\vec{v}]] \sim^\ap_0 t_2[u_0[\vec{v}],\ldots, u_{n-1}[\vec{v}]] = t_2[\vec{u}][\vec{v}], \]
whence $t_1[\vec{u}]\sim^\ap_m t_2[\vec{u}]$.   
\item[(3)] Suppose that $t_1\to t_1'$. It suffices to prove that $t_2$ reduces, that is, $t_2\to t_2'$ for some $t_2'\in \Lambda(n)$. Then, for every $\vec{u}\in \Lambda(0)^n$ we have  $t_1[\vec{u}]\sim^\ap_0 t_2[\vec{u}]$ by definition of $\sim^\ap_n$, and $t_1[\vec{u}]\to t_1'[\vec{u}]$ and $t_2[\vec{u}]\to t_2'[\vec{u}]$ because reductions respect substitution. Therefore $t_1'[\vec{u}] \sim^\ap_0 t_2'[\vec{u}]$ by (A1), which proves $t_1'\sim^\ap_n t_2'$ by definition of $\sim^\ap_n$.

To prove that $t_2$ reduces, suppose the contrary. There are two possible cases:

\medskip\noindent \underline{\emph{Case 1:}} $t_2$ is a $\lambda$-abstraction.\\
Since the term $t_1$ reduces, it is neither a variable nor a $\lambda$-abstraction. Therefore, for arbitrary $\vec{u}\in \Lambda(0)^n$, the term $t_1[\vec{u}]$ is not a $\lambda$-abstraction but $t_2[\vec{u}]$ is. Thus $t_1[\vec{u}]\not\sim^\ap_0 t_2[\vec{u}]$ and therefore $t_1\not\sim^\ap_n t_2$, a contradiction.

\medskip\noindent \underline{\emph{Case 2:}} $t_2=x\app s_1\app\cdots\app s_k$ for some $x\in V(n)$ and $s_1,\ldots, s_k\in \Lambda(n)$, $k\geq 0$. \\
Given $\lambda$-terms $s,t$ and a natural number $m>0$ we put
\begin{align*} s\xto{0} t\quad& :\iff \quad \text{$s=t$};\\
s\xto{m} t\quad&:\iff\quad \text{$s$ reduces to $t$ in exactly $m$ steps}.
\end{align*}
We shall prove that there exists $\vec{u}\in \Lambda(0)^n$ such that \[t_1[\vec{u}]\xto{m} \tilde{t_1}\qquad \text{and}\qquad t_2[\vec{u}]\xto{m}\tilde{t_2} \qquad \text{for some $m\geq 0$ and $\tilde{t_1},\tilde{t_2}\in \Lambda(0)$},\] where exactly one of the terms $\tilde{t_1}$ and $\tilde{t_2}$ is a $\lambda$-abstraction. Then $\tilde{t_1}\not\sim^\ap_0\tilde{t_2}$ by (A1) and (A2), whence $t_1[\vec{u}]\not\sim^\ap_0 t_2[\vec{u}]$ by $m$-fold application of (A1), and so $t_1\not\sim^\ap_n t_2$, a contradiction.

In order to construct $\vec{u}\in \Lambda(0)^n$ with the desired property, we consider several subcases:
  
\medskip\noindent \underline{\emph{Case 2.1:}} $t_1\xto{k} \ol{t_1}$ for some $\ol{t_1}$.

\medskip\noindent \underline{\emph{Case 2.1.1:}} $\ol{t_1}$ is a $\lambda$-abstraction. \\
Choose $\vec{u}$ such that $u_x\to u_x$ (e.g. $u_x=(\lambda y.y\app y)\app (\lambda y.y\app y)$). Then 
$t_2[\vec{u}] \xto{k} t_2[\vec{u}]$ and $t_2[\vec{u}]$ is not a $\lambda$-abstraction, while $t_1[\vec{u}]\xto{k} \ol{t_1}[\vec{u}]$ and $\ol{t_1}[\vec{u}]$ is a $\lambda$-abstraction.

\medskip\noindent \underline{\emph{Case 2.1.2:}} $\ol{t_1}$ is an application $\ol{t_{1,1}}\app \ol{t_{1,2}}$.\\
Choose $\vec{u}$ such that $u_x=\lambda x_1.\lambda x_2.\ldots \lambda x_k.\lambda y.y$. Then $t_2[\vec{u}]\xto{k} \lambda y.y$, while $t_1[\vec{u}]\xto{k} \ol{t_1}[\vec{u}]$ and $\ol{t_1}[\vec{u}]$ is not a $\lambda$-abstraction.

\medskip\noindent \underline{\emph{Case 2.1.3:}} $\ol{t_1}=x$. \\
Choose $\vec{u}$ such that $u_x=\lambda x_1.\lambda x_2.\ldots \lambda x_k.t$ where $t$ is is an arbitrary closed term that is not a $\lambda$-abstraction. Note that $u_x$ is a $\lambda$-abstraction: Since $t_1$ reduces, we have $t_1\neq x = \ol{t_1}$ and thus necessarily $k>0$. Thus $t_1[\vec{u}]\xto{k}\ol{t_1}[\vec{u}]=u_x$ and $u_x$ is a $\lambda$-abstraction, while $t_2[\vec{u}]\xto{k} t$ and $t$ is not a $\lambda$-abstraction.

\medskip\noindent \underline{\emph{Case 2.1.4:}} $\ol{t_1}=y$ for some variable $y\neq x$.\\
Choose $\vec{u}$ such that $u_x=\lambda x_1.\lambda x_2.\ldots \lambda x_k.\lambda y.y$ and $u_y$ is not a $\lambda$-abstraction. Then $t_2[\vec{u}]\xto{k} \lambda y.y$, while $t_1[\vec{u}]\xto{k} \ol{t_1}[\vec{u}]=u_y$ and $u_y$ is not a $\lambda$-abstraction.

\medskip\noindent \underline{\emph{Case 2.2:}} $t_1\xto{m} \ol{t_1}$ for some $m\in \{1,\ldots,k-1\}$ such that $\ol{t_1}$ does not reduce. 

\medskip\noindent \underline{\emph{Case 2.2.1:}} $\ol{t_1}$ is a $\lambda$-abstraction. \\
Choose $\vec{u}$ such that $u_x\to u_x$. Then $t_1[\vec{u}]\xto{m}\ol{t_1}[\vec{u}]$ and $\ol{t_1}[\vec{u}]$ is a $\lambda$-abstraction, while $t_2[\vec{u}]\xto{m} t_2[\vec{u}]$ and $t_2[\vec{u}]$ is not a $\lambda$-abstraction.

\medskip\noindent \underline{\emph{Case 2.2.2:}} $\ol{t_1} = y\app s_1'\ldots s_l'$ for some variable $y\neq x$ and terms $s_1',\ldots, s_l'$, $l\geq 0$.\\
Choose $\vec{u}$ such that $u_x\to u_x$ and $u_y=\lambda x_1.\lambda x_2.\ldots \lambda x_l.\lambda y.y$. Then $t_1[\vec{u}]\xto{m}\ol{t_1}[\vec{u}] \xto{l} \lambda y.y$ while ${t_2}[\vec{u}]\xto{m+l} {t_2}[\vec{u}]$ and ${t_2}[\vec{u}]$ is not a $\lambda$-abstraction.

\medskip\noindent \underline{\emph{Case 2.2.3:}} $\ol{t_1} = x\app s_1'\app\cdots\app s_l'$ for some $l> k-m$ and terms $s_1',\ldots, s_l'$.\\
Choose $\vec{u}$ such that $u_x=\lambda x_1.\lambda x_2.\cdots \lambda x_k.\lambda y.y$. Then $t_2[\vec{u}]\xto{k} \lambda y.y$, while \[t_1[\vec{u}]\xto{m} \ol{t_1}[\vec{u}] \xto{k-m} (\lambda x_{k-m+1}.\cdots \lambda x_k.\lambda y.y) \app s_{k-m+1}'[\vec{u}]\app\cdots\app 
  s_{l}'[\vec{u}]\] and $(\lambda x_{k-m+1}.\cdots \lambda x_k.\lambda y.y) \app s_{k-m+1}'[\vec{u}] \app\cdots\app 
  s_{l}'[\vec{u}]$ is not a $\lambda$-abstraction.

\medskip\noindent \underline{\emph{Case 2.2.4:}} $\ol{t_1} = x\app s_1'\app\cdots\app s_l'$ for some $l\leq k-m$ and terms $s_1',\ldots, s_l'$.\\
Choose $\vec{u}$ such that $u_x=\lambda x_1.\lambda x_2.\ldots \lambda x_k.t$ where $t$ is an arbitrary closed term that is not a $\lambda$-abstraction. Then \[t_2[\vec{u}]\xto{m+l} (\lambda x_{m+l+1}.\cdots \lambda x_k.t) \app s_{m+l+1}[\vec{u}] \app\cdots\app
  s_{k}[\vec{u}]\] and  $(\lambda x_{m+l+1}\cdots \lambda x_k.t) \app s_{m+l+1}[\vec{u}] \app\cdots\app 
  s_{k}[\vec{u}]$ is not a $\lambda$-abstraction (for $l=k-m$, this is just the term $t$), while
\[ t_1[\vec{u}] \xto{m} \ol{t_1}[\vec{u}] \xto{l} \lambda x_{l+1}.\cdots\lambda x_k.t \]
and  $\lambda x_{l+1}.\cdots\lambda x_k.t$ is a $\lambda$-abstraction since $l<k$.
\item[(6)] holds by symmetry to (3). 
\item[(4)] Suppose that $t_1=\lambda x.t_1'$. Then $t_2$ does not reduce (otherwise $t_1$ reduces by (6), a contradiction). Moreover, $t_2$ cannot be of the form $y\app s_1'\app\cdots\app s_l'$ where $y$ is variable and $s_1',\ldots,s_l'$ are terms. In fact, suppose the  contrary, and choose $\vec{u}\in \Lambda(0)^n$ such that $u_y$ is not a $\lambda$-abstraction. Then $t_1[\vec{u}] \not\sim^\ap_0 t_2[\vec{u}]$ since $t_1[\vec{u}]$ is a $\lambda$-abstraction and $t_2[\vec{u}]$ is not, contradicting $t_1\sim^\ap t_2$.

Thus $t_2=\lambda x.t_2'$ for $x=n$ and $t_2'\in \Lambda(n+1)$. Moreover, for every  $e\in \Lambda(n)$ and $\vec{u}\in \Lambda(0)^n$ we have
\[ t_1'[e/x][\vec{u}] = t_1'[\vec{u},e[\vec{u}]] = t_1'[\vec{u},x][e[\vec{u}]/x] \sim^\ap_0 t_2'[\vec{u},x][e[\vec{u}]/x] = t_2'[\vec{u},e[\vec{u}]] = t_2'[e/x][\vec{u}]   \] 
using (A2) and that $t_1[\vec{u}]\sim^\ap_0 t_2[\vec{u}]$ by definition of $\sim^\ap_n$. This proves $t_1'[e/x] \sim^\ap_n t_2'[e/x]$.
\item[(7)] holds by symmetry to (4).
\item[(5)] Suppose that $t_1=x\app s_1\app\cdots\app s_k$. Then $t_2$ does not reduce by (6) and is not a $\lambda$-abstraction by (7), so it must be of the form  $t_2=y\app s_1'\app\cdots\app s_m'$.
\item[(8)] holds by symmetry to (5).
\end{enumerate} 

\subsection*{Proof of \Cref{cor:cong}}
We only need to verify that our present setting satisfies the conditions of \Cref{th:main}:
\begin{enumerate} 
\item The category $\vcat$ is regular, being a presheaf topos.
\item The functor $\Sigma X = \Pt+\delta X + X\times X$ preserves
  reflexive coequalizers. In fact, $\delta$ is a left adjoint (with right adjoint $\llangle V+1,-\rrangle$, see~\cite{DBLP:conf/lics/FiorePT99}) and thus
  preserves all colimits. Moreover we use that reflexive
  coequalizers commute with finite products and coproducts in~$\Set$,
  hence also in $\vcat$ since limits and colimits are formed
  pointwise.
  
\item Finally, we show that the functor $B(X,Y)=\llangle X,Y\rrangle\times (Y+Y^X+1)$ preserves monos. Since monos in $\vcat$ are the componentwise injective natural transformations and thus stable under products and coproducts, it suffices to show that the functors  $(X,Y)\mapsto Y^X$ and $(X,Y) \mapsto \llangle X,Y \rrangle$ preserve monos.
The first functor preserves monos in the covariant component because $(-)^X\c \vcat\to \vcat$ is a right adjoint, and in the contravariant component because $Y^{(-)}\c (\vcat)^\opp\to \vcat$ is a right adjoint (with left adjoint $(Y^{(-)})^\opp\c \vcat\to (\vcat)^\opp$). 
The second functor preserves monos in the covariant component because $\llangle X,-\rrangle\c \vcat\to \vcat$ is a right adjoint. To see that it preserves monos in the contravariant component, suppose that $f\c X'\to X$ is an epimorphism in $\vcat$. Then $\llangle f,Y\rrangle\c \llangle X,Y\rrangle\to \llangle X',Y\rrangle$ is the natural transformation with components
\[ \llangle f,Y\rrangle_n\c \NT(X^n,Y)\to\NT((X')^n,Y), \qquad g\mapsto  g\comp f^n.\]
This map is clearly monic because $f$ is epic. Thus $\llangle f,Y\rrangle$ is monic in $\vcat$. 
\end{enumerate} 

\end{document}
